\documentclass{LMCS}

\def\dOi{10(3:19)2014}
\lmcsheading%
{\dOi}
{1--51}
{}
{}
{Mar.~14, 2012}
{Sep.~11, 2014}
{}

\ACMCCS{[{\bf Theory of computation}]: Semantics and
  reasoning---Program semantics---Categorical semantics} 

\keywords{corecursive algebra, corecursive monads, Bloom monads,
  iteration theory}

\usepackage{hyperref}
\usepackage{url}

\usepackage{amsmath, amsthm, amscd, amsfonts, amssymb}

\usepackage[all]{xy}
\SelectTips{cm}{}

\usepackage[final,layout=footnote]{fixme}

\newcommand{\takeout}[1]{\empty}

\theoremstyle{plain}
\newtheorem{theorem}{Theorem}[section]
\theoremstyle{definition}

\newtheorem{construction}[theorem]{Construction}
\numberwithin{equation}{section}

\def\card{\mathrm{card}\,}
\def\emp{\emptyset}
\def\o{\cdot}
\def\colim{\mathop{\mathrm{colim}}}

\begin{document}
%
%
\FXRegisterAuthor{sm}{asm}{SM}
%
%
\title[Corecursive Algebras, Corecursive Monads and Bloom Monads]{%
  Corecursive Algebras, Corecursive Monads \\and Bloom Monads}


\author[J.~Ad\'amek]{Ji\v{r}\'\i\ Ad\'amek\rsuper a}
\address{{\lsuper a}Institut f\"ur Theoretische Informatik, Technische Universit\"at Braunschweig, Germany}
\email{adamek@iti.cs.tu-bs.de}

\author[M.~Haddadi]{Mahdie Haddadi\rsuper b}
\address{{\lsuper b}Department of Mathematics, Statistics and Computer Science, Semnan University, Semnan, Iran}
\email{mahdiehaddadi7@gmail.com}

\author[S.~Milius]{Stefan Milius\rsuper c}
\address{{\lsuper c}Lehrstuhl f\"ur Theoretische Informatik, Friedrich-Alexander-Universit\"at Erlangen-N\"urnberg, Germany}
\email{mail@stefan-milius.eu}


\begin{abstract}
An algebra is called corecursive if from every coalgebra a unique coalgebra-to-algebra homomorphism exists into it. We prove that free corecursive algebras are obtained as coproducts of the terminal coalgebra (considered as an algebra) and free algebras. The monad of free corecursive algebras is proved to be the free corecursive monad, where the concept of corecursive monad is a generalization of Elgot's iterative monads, analogous to corecursive algebras generalizing completely iterative algebras. We also characterize the Eilenberg-Moore algebras for the free corecursive monad and call them Bloom algebras.
\end{abstract}

\maketitle

%
%

\section{Introduction}

The study of structured recursive definitions is fundamental in many
areas of computer science. This study can use algebraic methods
extended by suitable recursion concepts. One such example are
completely iterative algebras: these are algebras in which recursive
equations with parameters have unique solutions, see \cite{m_cia}. In
the present paper we study corecursive algebras.  These are
$H$-algebras for a given endofunctor $H$ in which recursive equations
without parameters have unique solutions. Equivalently, for every
coalgebra there exists a unique coalgebra-to-algebra morphism. The
dual concept, recursive coalgebra, was introduced by G.~Osius in
\cite{g}, and for endofunctors weakly preserving pullbacks P.~Taylor
proved that this is equivalent to being parametrically recursive, see
\cite{t}. In the dual situation, since weak preservation of pushouts
is rare, the concepts of corecursive algebra and completely iterative
algebra usually do not coincide. The former was studied by
V.~Capretta, T.~Uustalu and V.~Vene \cite{cuv2}, and various
counter-examples demonstrating e.g. the difference of the two concepts
for algebras can be found there. In the present paper we contribute to
the development of the mathematical theory of corecursive
algebras. The goal is to eventually arrive at a useful body of results
and constructions for these algebras. A major ingredient of any theory
of algebraic structures is the study of how to freely endow an object
with the structure of interest. So the main focus of the present paper
are corecursive $H$-algebras freely generated by an object $Y$. Let
$FY$ denote the free $H$-algebra on $Y$ and let $T$ be the terminal
$H$-coalgebra (which, due to Lambek's Lemma, can be regarded as an
algebra). We prove that the coproduct of these two
algebras $$MY=T\oplus FY$$ is the free corecursive algebra on
$Y$. Here $\oplus $ is the coproduct in the category of
$H$-algebras. For example for the endofunctor $HX=X\times X$ the
algebra $MY$ consists of all (finite and infinite) binary trees with
finitely many leaves labelled in $Y$.

We also introduce the concept of a corecursive monad.  This is a
weakening of completely iterative monads of C.~Elgot, S.~Bloom and R.~Tindell \cite{ebt} 
analogous to corecursive algebras as a weakening of completely
iterative ones.  The monad $Y\mapsto MY$ of free corecursive algebras
is proved to be corecursive, indeed, this is the free corecursive
monad generated by $H$. For endofunctors of {\bf Set} we also prove
the converse: whenever $H$ generates a free corecursive monad, then it
has free corecursive algebras (and the free monad is then given by the
corresponding adjunction).

We characterize the Eilenberg-Moore algebras for the free corecursive monad:  these are $H$-algebras in which  every recursive equation without parameters has a solution (not necessarily unique), and which allow a functorial choice of solutions. We call these algebras \emph{Bloom algebras}; they are analogous to the complete Elgot algebras of \cite{amv3} where the corresponding monad was the free completely iterative monad on $H$.

We also study the finitary versions of our concepts. An algebra $A$ is
called finitary corecursive if all coalgebras on finitely presentable
objects have a unique coalgebra-to-algebra morphism in to $A$. And
finitary corecursive monads are defined analogously. Every finitary
endofunctor $H$ is proved to generate a free finitary corecursive
monad $\mathbb{M}_H$. We form the free strict functor $H_\bot=H+1$ on
$H$ and obtain a monad $\mathcal{M}^*$ on the category of finitary
functors given by
$$\mathcal{M}^*(H)=\mathbb{M}_{H_\bot}.$$ 
The Eilenberg-Moore algebras for $\mathcal{M}^*$ are called Bloom
monads. They correspond to iteration theories of Bloom and \'Esik:
recall from \cite{amv_what} that the latter are precisely the
Eilenberg-Moore algebras for the free-iteration-theory monad. Bloom
monads are monads $\mathbb{S}$ equipped with an operation $\dag$
assigning to every finitary non-parametric equation morphism a
solution in free algebras for $\mathbb{S}$. This operation satisfies
precisely the equational properties that non-parametric iteration in
Domain Theory satisfies. We list some of those equational
properties. The question whether our list is complete is open.

This paper is a revised and extended version of the conference paper
\cite{ahm11}. Here we added all technical details and proofs, and our
discussion of the equational properties of Bloom monads and their
properties in new.

\subsection*{Acknowledgments}
We are grateful to Zolt\'an \'Esik for a substantial contribution to the discussion of equations in Bloom monads. And to Paul Levy who suggested that  Proposition \ref{Bloom cat=slice cat} holds.

\section{Corecursive Algebras}
The following definition is the dual of the concept introduced by G. Osius in \cite{g} and studied by P. Taylor \cite{t,t2}. We assume throughout the paper that  a category $\mathcal A$ and an endofunctor $H:\mathcal A\rightarrow \mathcal A$ are given. We denote by $\mathsf{Alg}\, H$ the category of algebras $a:HA\rightarrow A$ and homomorphisms, and by $\mathsf{Coalg}\, H$ the category of coalgebras $e:X\rightarrow HX$ and homomorphisms. A coalgebra-to-algebra morphism from the latter to the former  is a morphism $f:X\rightarrow A$ such that $f=a\cdot Hf\cdot e$.

\begin{defi}\label{def of cor.alg}
An algebra $a:HA\rightarrow A$ is called {\it corecursive} if for every coalgebra $e:X\rightarrow HX$ there exists a unique coalgebra-to-algebra homomorphism $e^\dagger:X\rightarrow A$. That is, the square
\begin{equation}\label{eq:sol}
  \vcenter{
    \xymatrix{
      X\ar[r]^{e^\dagger}\ar[d]_{e}&A\\
      HX\ar[r]_{He^\dagger}&HA\ar[u]_{a}
    }
  }
\end{equation}
commutes. We call $e$ an {\it equation morphism} and $e^\dagger$ its {\it solution}.
\end{defi}
\begin{rem}
 For an endofunctor on {\bf Set}, we can view $e:X\rightarrow HX$ as a system of recursive equations using variables from the set $X$, and $e^{\dagger}:X\rightarrow A$ is the solution of the system. We illustrate this on classical $\Sigma$-algebras. These are the algebras for the polynomial set functor
$$H_\Sigma X=\coprod_{\sigma\in \Sigma} X^n$$
where $n$ is the arity of $\sigma$. For every set $X$ (of recursion variables) and every system of mutually recursive equations
$$x=\sigma(x_1,\ldots,x_n),$$
one for every $x\in X$, where $\sigma\in \Sigma$ has arity $n$ and $x_i\in X$, we get the corresponding coalgebra
$$
\xymatrix@1{
  e:X\ar[r] & H_{\Sigma}X;
}
\qquad 
x\mapsto(x_1,\ldots,x_n)
\qquad \textrm{in the $\sigma$-summand $X^n$.}
$$
The square~(\ref{eq:sol}) tells us that the substitution of $e^\dagger(x)$ for $x\in X$ makes the formal equations $x=\sigma(x_1,\ldots,x_n)$ identities in $A$:
$$e^\dagger (x)=\sigma^A(e^\dagger(x_1),\ldots,e^\dagger(x_n)).$$
\end{rem}
\begin{exa}\label{exam of core}\hfill
  \begin{enumerate}
  \item  In~\cite{cuv2} this concept of corecursive algebras is studied and compared with a number of related concepts. A concrete example of a corecursive algebra from that paper, for the endofunctor $HX=E\times X\times X$ on {\bf Set}, is the set $E^\infty$ of all streams. The operation $a:E\times E^\infty\times E^\infty\rightarrow E^\infty$ is given by $a(e,u,v)$ having head $e$ and continuing by the merge of $u$ and $v$.

\takeout{
(2) In a category with finite products consider algebras on one binary operation, i.e., $HX=X\times X$. An algebra $a:A\to A\times A$ is corecursive iff it has a unique idempotent (global) element. That is, a unique $i:1\to A$ with $a\cdot (i\times i)=a$.

Indeed, this is necessary since $i$ is the solution of $e=id :1\to
1\times 1$. And it is sufficient: given an equation morphism $e:X\to
X\times X$, the morphism $X\stackrel{!}\to 1\stackrel{i}\to A$ is
easily seen to be the unique solution of $e$ in $A$.
}

\item If $H$ has a terminal coalgebra $\tau:T\rightarrow HT$, then by Lambek's Lemma $\tau$ is invertible and the resulting algebra $\tau^{-1}:HT\rightarrow T$ is corecursive. In fact, this is the initial corecursive algebra, that is, for every corecursive algebra $(A,a)$  a unique algebra homomorphism  from $(T,\tau^{-1})$ exists, see the dual of~\cite[Proposition~2]{cuv2}. There also the converse is proved (dual of Proposition 7), that is, if the initial corecursive algebra exists, then it is a terminal coalgebra (via the inverse of the algebra structure).

\item The trivial terminal algebra $H1\rightarrow 1$, where $1$ is the terminal object in $\mathcal A$, is clearly corecursive.

\item If $a:HA\rightarrow A$ is a corecursive algebra, then so is $Ha:HHA\rightarrow HA$, see~\cite[Proposition~6]{cuv2}. We generalize this in Lemma \ref{generalized} below.

\item Combining (3) and (4) we conclude that the terminal $\omega^{op}$-chain
$$
\xymatrix{
1&H1\ar[l]_-a&HH1\ar[l]_-{Ha}&\ldots\ar[l]_-{HHa}
}
$$
consists of corecursive algebras. Indeed, the continuation to $H^i1$
for all ordinals (with $H^i1=\lim_{k\leq i} H^k1$ for all limit
ordinals) also yields corecursive algebras. This follows from the
following.
\end{enumerate}
\end{exa}

\begin{prop}\label{lim of core is core}
Let $\mathcal A$ be a complete category. Then corecursive algebras are closed under limits in $\mathsf{Alg}\, H$. Thus, limits of corecursive algebras are formed on the level of $\mathcal A$.
\end{prop}

\begin{proof}
It is easy to verify that limits in $\mathsf{Alg}\, H$ are formed on the level of $\mathcal A$. Let us prove that the product of corecursive algebras is  corecursive. The proof for general limits is analogous.

Let $(A,a)$ be  the product of corecursive algebras $(A_i,a_i)$, with projections $p_i:A\rightarrow A_i$. For every coalgebra $e:X\rightarrow HX$ we have the unique coalgebra-to-algebra morphism $e^\dagger _i:X\rightarrow A_i$, for all $i\in I$, and the morphism $e^\dagger =\langle e_i^\dagger\rangle :X\rightarrow A=\prod_{i\in I}A_i$ is a coalgebra-to-algebra morphism. Indeed, for every $i\in I$, the diagram
$$
\xymatrix{
X\ar[r]^{e^\dagger}\ar[d]_{e}&A\ar[r]^{p_i}&A_i\ar@{<-}`u[l]`[ll]_{e^\dagger_i}[ll]\\
HX\ar[r]_{He^\dagger}&HA\ar[u]_a\ar[r]_{Hp_i}&HA_i\ar[u]_{a_i}\ar@{<-}`d[l]`[ll]^{He^\dagger_i}[ll]
}
$$
commutes, except perhaps for the left hand inner square; but this
suffices to establish the desired commutativity of the left hand
square.  Since all $A_i$ are corecursive, the uniqueness of
$e^\dagger$ follows from the observation that there is a one-one
correspondence between solutions $s:X\rightarrow A$ of $e$ in $A$ and
families of solutions $s_i:X\rightarrow A_i$ of $e$ in $A_i$, for all
$i\in I$.
\end{proof}

In the following we write $\mathsf{inl}:X\rightarrow X+Y$ and $\mathsf{inr}:Y\rightarrow X+Y$ for the injections of a coproduct.

\begin{lem}\label{generalized}
Let $(A,a)$ be an algebra and $f:B\rightarrow A$  a morphism. Then $(A,a)$ is a corecursive algebra if and only if the algebra
$$\xymatrix{\overline a\equiv H(HA+B)\ar[r]^-{H[a,f]}& HA\ar[r]^-{\mathsf{inl}}& HA+B}$$
is corecursive.
\end{lem}

\begin{proof}
Let $(A, a)$ be a corecursive algebra and $e:X\rightarrow HX$ be an equation morphism. Then there is a unique solution:
\begin{equation}\label{1}
\vcenter{\xymatrix{
X\ar[r]^{e^\dagger}\ar[d]_{e}&A\\
HX\ar[r]_{He^\dagger}&HA\ar[u]_a
}}
\end{equation}
Now inspection of the following commutative diagram shows that $\mathsf{inl}\cdot He^\dagger\cdot e:X\rightarrow HA+B$ is a solution of $e$ in $HA+B$.
\begin{equation}\label{*}
\vcenter{\xymatrix@C+1pc{
X\ar[r]^e\ar[d]_-{e}&HX\ar[r]^-{He^\dagger}\ar[d]^-{He}&HA\ar[r]^-{\mathsf{inl}}&HA+B\\
HX\ar@{=}[ur]\ar[r]_-{He}&HHA\ar[r]_-{HHe^\dagger}&HHA\ar[u]^{Ha}\ar[r]_-{H\mathsf{inl}}&H(HA+B)\ar[lu]^{H[a,f]}\ar[u]_{\overline a}
}}
\end{equation}
Indeed, commutativity of the middle rectangle follows from Diagram (\ref{1}), the lower triangle on the right is trivial and the upper triangle is the definition of $\overline a$.
To show the uniqueness of the solution, suppose that $s:X\rightarrow HA+B$ is a solution for $e$, so we have the following commutative diagram:
$$
\xymatrix@C+1pc{
X\ar[r]^-{s}\ar[d]_e&HA+B\ar[r]^-{[a,f]}&A\\
HX\ar[r]_-{Hs}&H(HA+B)\ar[u]^{\mathsf{inl}\cdot H[a,f]}\ar[r]_-{H[a,f]}&HA\ar[u]_{a}
}
$$
Since the solution $e^\dagger$ in $A$ is unique, $e^\dagger=[a,f]\cdot s$ and hence
$$\mathsf{inl}\cdot He^\dagger \cdot e = \mathsf{inl}\cdot H[a,f]\cdot Hs\cdot e = s.$$

Conversely, let $(HA+B, \mathsf{inl}\cdot H[a,f])$ be a corecursive algebra and $e:X\rightarrow HX$ be an equation morphism. So there exists a unique solution $e^\dagger $ of $e$ in $HA+B$, and hence we have the above commutative diagram with $e^\dagger$ in lieu  of $s$. That is $[a,f]\cdot e^\dagger$ is a solution of $e$ in the algebra  $(A,a)$. To show uniqueness suppose that $s:X\rightarrow A$ is a solution of $e$, that is $s=a\cdot Hs\cdot e$. Then we have Diagrams (\ref{1}) and (\ref{*}) with the morphism $s$ in lieu of $e^\dagger$. So, by uniqueness of solution in the corecursive algebra $(HA+B, \mathsf{inl}\cdot H[a,f])$, we have   $\mathsf{inl}\cdot Hs\cdot e=e^\dagger$, and hence $s=a\cdot Hs\cdot e=[a,f]\cdot\mathsf{inl}\cdot Hs\cdot e=[a,f]\cdot e^\dagger$.
\end{proof}

\begin{exa}
  \label{ex:binary}
Binary algebras: For $HX=X\times X$, every algebra  (given by the binary operation ``$*$" on a set $A$) which is corecursive has a unique {\it idempotent} $i=i*i$. This is the solution of the recursive equation $$x=x*x$$
expressed by the isomorphism $e:1\stackrel{\sim}\rightarrow 1\times
1$. Moreover the idempotent is {\it completely factorizable}, where
the set of all completely factorizable elements is defined to be the largest subset of $A$ such that every element $a$ in it can be factorized as $a=b*c$, with $b,c$ completely factorizable. The corecursiveness of $A$ implies  that no other element but $i$ is completely factorizable: consider the system of recursive equations
\begin{align}
x_\epsilon&=x_0*x_1,&
x_0&=x_{00}*x_{01}, &
&\cdots&
x_w&=x_{w 0}*x_{w 1}, &
\cdots\notag
\end{align}
for all finite binary words $w$. Every completely factorizable element $a$ provides a solution $e^\dagger$ with $e^{\dagger}(x_\epsilon )=a$. Since solutions are unique, $a=i$.

Conversely, every binary algebra $A$ with an idempotent $i$ which is the only completely factorizable element is corecursive. Indeed, given a morphism $e:X\rightarrow X\times X$, the constant map $e^{\dagger}:X\rightarrow A$ with value $i$ is a coalgebra-to-algebra morphism. Conversely, if $e^\dagger$ is a coalgebra-to-algebra morphism, then for every $x\in X$ the element $e^\dagger (x)$ is clearly completely factorizable. Therefore, $e^\dagger (x)=i$.
\end{exa}

\begin{rem}\label{cia}
Recall the concept of {\it completely iterative algebra} ({\it cia} for short) from \cite{m_cia}: it is an algebra $(A,a)$ such that for every ``flat equation" morphism $e:X\rightarrow HX+A$ there exists a unique solution, i.e. a unique morphism $e^\dagger$ such that the square
 $$
\xymatrix{
X\ar[r]^{e^\dagger}\ar[d]_{e}&A\\
HX+A\ar[r]_{He^\dagger+A}&HA+A\ar[u]_{[a,A]}
}
$$
commutes. This is obviously stronger than corecursiveness because every coalgebra $e:X\rightarrow HX$ yields a flat equation morphism ${\mathsf{inl}}\cdot e:X\rightarrow HX+A$. Then solutions are determined uniquely. Thus, for example, in the category of complete metric spaces with distance less than one and nonexpanding functions, all algebras for contracting endofunctors (in the sense of P. America and J. Rutten \cite{america+rutten}) are corecursive, because, as proved in \cite{m_cia}, they are cia's. Here is a concrete example: $HX=X\times X$ equipped with the metric taking $1/2$ of the maximum of the two distances is contracting. Thus every binary algebra whose operation is contracting is corecursive.
\end{rem}

\begin{exa}
The endofunctor $HX=X\times X$ has many corecursive algebras that are not cia's. For example the algebra $A$ of all binary trees with finitely many leaves. The operation is tree-tupling and the only completely factorizable tree is the complete binary tree $t$. Thus, $A$ is corecursive. However, if $a\in A$ denotes the root-only tree, then the system of recursive equations
\begin{align}
x&=x*y\notag\\
y&=a\notag
\end{align}
does not have a solution in $A$ (because the tree corresponding to $x$ has infinitely many leaves). Thus $A$ is not a cia.
\end{exa}

\begin{lem}\label{homos are solution prese}
Every homomorphism $h:(A,a)\rightarrow (B,b)$ in $\mathsf{Alg}\, H$ with $(A,a)$ and $(B,b)$ corecursive preserves solutions. That is,  given a coalgebra $e:X\rightarrow HX$ with a solution $e^{\dagger}:X\rightarrow A$ in the domain algebra, then $h\cdot e^\dagger:X\rightarrow B$ is the solution in the codomain one.
\end{lem}

\begin{proof}
This follows from the diagram
\[
\xymatrix{
X\ar[r]^{e^\dagger}\ar[d]_{e}&A\ar[r]^h&B\\
HX\ar[r]_{He^\dagger}&HA\ar[u]_{a}\ar[r]_{Hh}&B\ar[u]_b 
}
\vspace*{-22pt}
\]
\end{proof}

\noindent We thus consider corecursive algebras as a full subcategory $\mathsf{Alg}_C\, H$ of $\mathsf{Alg}\, H$. We obtain a forgetful functor
\[\begin{array}{cccccc}
\mathsf{Alg}_C\,H&\rightarrow &\mathcal A\qquad
(A,a)&\mapsto&A
\end{array}\]

In Section~\ref{sec:4} we prove that this forgetful functor has a left adjoint, that is, free corecursive algebras exist, if and only if a terminal coalgebra $T$ exists and every object $Y$ generates a free algebra $FY$ (i.e., the forgetful functor $\mathsf{Alg}\, H\rightarrow \mathcal A$ has a left adjoint). Our result holds for example for all set functors, and for them the formula for the free corecursive algebra is $T\oplus FY$, where $\oplus$ is the coproduct in $\mathsf{Alg}\, H$.

Recall from \cite{GU} that given an infinite cardinal number $\lambda$, a functor is called {\it $\lambda$-accessible} if it preserves $\lambda$-filtered colimits. An object $X$ whose hom-functor $\mathcal A(X,-)$ is $\lambda$-accessible is called {\it $\lambda$-presentable}. A category $\mathcal A$ is {\it locally $\lambda$-presentable} if it has
\begin{enumerate}[label=\({\alph*}]
\item colimits, and
\item a set of $\lambda$-presentable objects whose closure under $\lambda$-filtered colimits is all of $\mathcal A$. 
\end{enumerate}
A category $\mathcal A$ is called \emph{locally presentable} (and a functor $F$ \emph{accessible}, resp.) if there exists some $\lambda$ such that $\mathcal A$ is locally $\lambda$-presentable (and $F$ $\lambda$-accessible, resp.). 

For a $\lambda$-accessible endofunctor $H$, the category $\mathsf{Alg}\, H$ is also locally $\lambda$-presentable, see \cite{ar}. For corecursive algebras we have:

\begin{prop}\label{dirlim of cor is cor}
Let $\mathcal A$ be a locally presentable category. Then for every accessible endo\-functor $H$, the category $\mathsf{Alg}_C\, H$ of corecursive algebras is locally presentable.
\end{prop}

\begin{proof}
Choose an uncountable cardinal number $\lambda$ such that $H$ preserves $\lambda$-filtered colimits and $\mathcal A$ is locally $\lambda$-presentable. Then $\lambda$-filtered colimits in $\mathsf{Alg}\,  H$ are clearly formed on the level of $\mathcal A$. And as proved in \cite{ap}, every coalgebra is a $\lambda$-filtered colimit of $\lambda$-presentable coalgebras, and these are precisely the coalgebras carried by $\lambda$-presentable objects in $\mathcal A$.

By the Reflection Theorem (see \cite[Corollary to Theorem 2.48]{ar}), in order to establish that the full subcategory $\mathsf{Alg}_C\, H$ of $\mathsf{Alg}\, H$ is locally $\lambda$-presentable, it suffices to see that it is closed in $\mathsf{Alg}\, H$ under limits and $\lambda$-filtered colimits. For limits see Proposition  \ref{lim of core is core}, and now we prove that $\lambda$-filtered colimits of corecursive algebras in $\mathsf{Alg}\, H$ are corecursive. Indeed, let $(A_t, a_t)_{t\in T}$ be a $\lambda$-filtered diagram with colimit $k_t:(A_t,a_t)\rightarrow (C,c)$. For every coalgebra $e:X\rightarrow HX$, a solution $e^\dagger:X\rightarrow A_t$ exists in $(A_t,a_t)$ and since $k_t$ is a homomorphism, $k_t\cdot e^\dagger$ is a solution in $C$, see Lemma \ref{homos are solution prese}.

To prove that solutions are unique, assume first that $X$ is $\lambda$-presentable in $\mathcal A$. For every solution $e^\dagger:X\rightarrow C$ there exists a $t\in T$ such that $e^\dagger$ factorizes through $k_t$ as follows

$$
\xymatrix{
X\ar[r]^{e^\dagger}\ar[d]_e&C&A_t\ar[l]_{k_t}\ar@{<-}`u[l]`[ll]_-s[ll]\\
HX\ar[r]_{He^\dagger}&HC\ar[u]_c&HA_t\ar[u]_{a_t}\ar[l]^{Hk_t}\ar@{<-}`d[l]`[ll]^-{Hs}[ll]\\
}
$$
The morphism $k_t$ merges $s$ and $a_t\cdot Hs \cdot e$:

\begin{align*}
k_t\cdot(a_t\cdot Hs\cdot e)&=c\cdot Hk_t\cdot Hs\cdot e\\
&=c\cdot He^\dagger \cdot e\\
&=e^\dagger \\
&=k_t\cdot s.
\end{align*}
Consequently, since $k_t$ is a colimit morphism of a $\lambda$-filtered colimit, there exists an object $t'\in T$ and a connecting morphism $u:A_t\rightarrow A_{t'}$ which also merges $s$ and $a_t\cdot Hs\cdot e$, that is $$u\cdot s=u\cdot a_t\cdot Hs\cdot e=a_{t'}\cdot Hu\cdot Hs\cdot e = a_{t'} \cdot H(u\cdot s) \cdot e.$$
This last equation proves that $u\cdot s$ is a solution of $e$ in $A_{t'}$, thus $u\cdot s$ is uniquely determined. Hence, $e^\dagger $ is uniquely determined from $e^\dagger=k_t\cdot s=k_{t'}\cdot u\cdot s$.

Next let $X$ be arbitrary. Express $(X,e)$ in the category of
coalgebras as a $\lambda$-filtered colimit of coalgebras $(X_i, e_i)$
with $X_i$ $\lambda$-presentable in $\mathcal A$. Let $x_i:X_i\rightarrow X$ be
the corresponding colimit cocone. For every solution
$e^\dagger:X\rightarrow C$ each $e^\dagger \cdot x_i$ is a solution of
$e_i$ since we have the following commutative diagram
$$
\xymatrix{
X_i\ar[r]^{x_i}\ar[d]_{e_i}&X\ar[r]^{e^\dagger}\ar[d]^e&C\\
HX_i\ar[r]_{Hx_i}&HX\ar[r]_{He^\dagger}&HC\ar[u]_c
}
$$
Thus, $e^\dagger \cdot x_i$ is uniquely determined by the previous case. Since the cocone of all $x_i$'s is collectively epic, this proves that $e^\dagger$ is uniquely determined.
\end{proof}

\begin{rem}
  We obtain from Proposition~\ref{dirlim of cor is cor} that for an uncountable cardinal number $\lambda$, if $H$ is $\lambda$-accessible and $\mathcal A$ locally $\lambda$-presentable, then so is $\mathsf{Alg}_C\, H$. And for $\lambda = \omega$, i.e.~$H$ is finitary on the locally finitely presentable category $\mathcal A$, we have that $\mathsf{Alg}_C\, H$ is locally $\aleph_1$-presentable.
\end{rem}

\section{Bloom Algebras}

In the case of iterative algebras, it was proved in \cite{amv_atwork}
that every finitary functor $H$ of $\mathcal A$ has a free iterative
algebra on every object of $\mathcal A$, and the resulting monad $\mathbb{R}$ on $\mathcal A$ is a
free iterative monad. The next step was a characterization of the
Eilenberg-Moore algebras for $\mathbb R$ that were called Elgot
algebras~\cite{amv3}. An Elgot algebra has for every finitary flat
equation $e$ a solution $e^\dagger$, but not necessarily
unique. Instead, Elgot algebras are equipped with a solution operation
$e\mapsto e^\dagger$ satisfying some ``natural" axioms.

In the present section we take the corresponding step for corecursive algebras. We introduce Bloom algebras as algebras equipped with an operation assigning to every coalgebra $e$ a solution $e^\dagger$ which forms a functor. Later we prove that Bloom algebras are (analogously to Elgot algebras) precisely the Eilenberg-Moore algebras for the free corecursive monad, see Theorems~\ref{4.13} and~\ref{M is free core monad}.

\begin{defi}
A {\it Bloom algebra} is a triple $(A,a,\dagger)$ where
$a:HA\rightarrow A$ is an $H$-algebra and $\dagger$ is an operation
assigning to every coalgebra $e:X\rightarrow HX$ a
coalgebra-to-algebra homomorphism $e^\dagger:X\rightarrow A$ so that
$\dagger $ is functorial. This means that we obtain a functor
$$\dagger:\mathsf{Coalg}\, H\rightarrow {\mathcal A}/A.$$
More explicitly, given a coalgebra homomorphism
$$
\xymatrix{
X\ar[r]^{e}\ar[d]_{h}&HX\ar[d]^{Hh}\\
X'\ar[r]_{f}&HX'
}
$$\enlargethispage{\baselineskip}
 the following triangle commutes
$$
\xymatrix{
X\ar[rr]^h\ar[dr]_{e^\dagger}&&X'\ar[ld]^{f^\dagger}\\
&A
}
$$
\end{defi}

\begin{exa}\label{ex 3}\hfill
  \begin{enumerate}[label=\({\alph*}]
  \item Every corecursive algebra is a Bloom algebra. Indeed,
    functoriality easily follows from the uniqueness of solutions due
    to the diagram
    $$
    \xymatrix{
      X\ar[r]^{h}\ar[d]_{e}&HX\ar[d]^{Hf}\ar[r]^{f^\dagger}&A\\
      HX\ar[r]_{Hh}&HX'\ar[r]_{Hf^\dagger}&HA\ar[u]_a
    }
    $$

  \item A unary algebra $a: A \to A$ ($H = Id$) is a Bloom algebra iff
    $a$ has a fixpoint, i.\,e., a morphism $t: 1 \to A$ with $a\cdot t
    = t$. More precisely:
    \begin{enumerate}[label=(\arabic*)]
    \item Given a fixpoint, then $(A, a, \dagger)$ is a Bloom algebra
      where $e^\dagger = t \cdot !$ for the unique morphism $!: X \to 1$.
    \item Given a Bloom algebra $(A,a,\dagger)$, then $id_1^\dagger: 1
      \to A$ is a fixpoint of $a$.
    \end{enumerate}
  \item Let $\mathcal A$ have finite products. An algebra $a:A\times
    A\rightarrow A$ for $HX=X\times X$ is a Bloom algebra if and only
    if it has an idempotent global element, that is $i:1\rightarrow A$
    satisfying $a\cdot(i\times i)=i$ (recall that $1\times 1=1$). More
    precisely:

    \begin{enumerate}[label=(\arabic*)]
    \item Given an idempotent $i$, we have a Bloom algebra $(A,a,\dagger)$, where $\dagger $ is the constant function with value $e^\dagger=i\cdot !$.

    \item Given a Bloom algebra $(A,a,\dagger)$, there exists an idempotent $i$ such that $\dagger$ is the constant function with value $e^\dagger=i\cdot !$.
    \end{enumerate}
    Compare this with Example \ref{ex:binary}. In particular every group, considered as a binary algebra in {\bf Set}, is thus a Bloom algebra in a unique sense. But no nontrivial group is corecursive.

  \item Every continuous algebra is a Bloom algebra if we define
    $e^\dagger$ to be the least solution of $e$. More detailed, let
    $H$ be a locally continuous endofunctor of the category
    $\mathsf{CPO}$ of complete ordered sets (i.\,e., partially ordered
    sets with a least element $\bot$ and with joins of
    $\omega$-chains). For every $H$-algebra $(A,a)$ and every equation
    morphism $e:X\rightarrow HX$, we can define in $\mathsf{CPO}(X,A)$
    a function $e^\dagger :X\rightarrow A$ as a join of the sequence
    $e_n^\dagger$ defined by $e_0^\dagger=\mathsf{const}_\bot$ and
    $e^\dagger_{n+1}=a\cdot He^\dagger_n\cdot e$. Then the least
    solution of $e$ is $e^\dagger =\bigvee_{n<\omega} e_n^\dagger$ and
    $(A,a,\dagger)$ is a Bloom algebra. Example~(b) demonstrates that
    this need not be corecursive.

  \item Every product of Bloom algebras is a Bloom algebra. We define $e^\dagger =\langle e_i^\dagger \rangle$ as in the proof of Proposition \ref{lim of core is core}. More generally: every limit of Bloom algebras is a Bloom algebra.

  \item Every complete Elgot algebra in the sense of
    \cite{amv_classes} is a Bloom algebra.
  \end{enumerate}
\end{exa}

\begin{defi}\label{preserv solu}
By a {\it homomorphism} of Bloom algebras from $(A,a,\dagger)$ to $(B,b,\ddag)$ is meant an algebra homomorphism $h:(A,a)\rightarrow (B,b)$ preserving solutions, that is, for every coalgebra $e:X\rightarrow HX$ the triangle
$$
\xymatrix{
A\ar[rr]^{h}&&B\\
&X\ar[ur]_{e^\ddag}\ar[ul]^{e^\dag}
}
$$
commutes. We denote by $\mathsf{Alg}_B\, H$ the corresponding category of Bloom algebras.
\end{defi}

\begin{prop}\label{Bloom cat=slice cat}
Let $(T,\tau)$ be a terminal coalgebra for $H$. The category of Bloom algebras for $H$ is isomorphic to the slice category $(T,\tau^{-1})/\mathsf{Alg}\, H$.
\end{prop}

\begin{proof}
Let us, for a coalgebra $(X,e)$, denote the unique coalgebra homomorphism from $X$ to $T$ by $e^\sharp:X\rightarrow T$.
We shall define two functors between $\mathsf{Alg}_B\,H$ and $\mathcal S=(T,\tau^{-1})/\mathsf{Alg}\, H$ and show that they are mutually  inverse.

(a) From Bloom algebras to the slice category $\mathcal{S}$: given a Bloom algebra $(A,a,\dagger)$ we form the solution $\tau^\dagger:T\rightarrow A$ which clearly is an object $(A,\tau^\dagger)$ of $\mathcal S$. For a homomorphism $h:(A,a,\dagger)\rightarrow (B,b,\ddagger)$ of Bloom algebras, we clearly have a morphism $h:(A,\tau^\dagger)\rightarrow (B,\tau^{\ddagger})$ of $\mathcal S$, since $h$ is solution preserving. This defines a functor from $\mathsf{Alg}_B\,H$ to $\mathcal S$.

(b) From $\mathcal S$ to Bloom algebras: Suppose we are given an object $(A,h)$ in $\mathcal S$, that is, an algebra homomorphism:
$$
\xymatrix{
T\ar[d]_h&HT\ar[l]_{\tau^{-1}}\ar[d]^{Hh}\\
A&HA\ar[l]^a
}
$$
We define for every  $e:X\rightarrow HX$ its dagger as $e^\dagger = h\cdot e^\sharp$. This is functorial; indeed, for every coalgebra homomorphism $k:(X,e)\rightarrow (Y,f)$ we have $f^\sharp\cdot k=e^\sharp$ by unicity of the universal property of the terminal coalgebra $(T,\tau)$, thus
\begin{align}
f^\dagger\cdot h&=h\cdot f^\sharp\cdot k\notag\\
&=h\cdot e^\sharp\notag\\
&=e^\dagger\notag
\end{align}
 In addition, every morphism $m:(A,h)\rightarrow (B,h')$ of $\mathcal
 S$ is a homomorphism of Bloom algebras $(A,a,\dagger)\rightarrow (B,b,\ddagger)$:
\begin{align}
m\cdot e^\dagger&=m\cdot h\cdot e^\sharp\notag&\text{by definition of $\dagger$\qquad}\\
&=h'\cdot e^\sharp\notag&\text{$m$ is a morphism in $\mathcal S$}\\
&=e^\ddag\notag&\text{by definition of $\ddag$\qquad}
\end{align}
That this gives a functor from $\mathcal S$ to $\mathsf{Alg}_B\,H$ is immediate.

(c) The two functors above  are mutually inverse. Indeed, it suffices to show that we have  a bijection on the level of objects, since both functors    are  the identity maps on morphisms. So for $(A,h)$ in $\mathcal S$ we form first $(A,a,\dagger)$ as in (b) and then $(A,\tau^\dagger)$ as in (a) and we have $\tau^\dagger=h\cdot\tau^\sharp=h$, as $\tau^\sharp$ is the identity (being the unique coalgebra homomorphism from $(T,\tau)$ to itself). Finally, given a Bloom algebra $(A,a,\ddagger)$ we first form $(A,\tau^\ddagger)$ as in (a) and then $(A,a,\dagger)$ as in (b). Then we have
\begin{align}
e^\dagger&=\tau^\ddagger\cdot e^\sharp\notag&\text{by definition of $\dagger$ in (b)}\\
&=e^\ddag\notag&\text{by functoriality of $\ddagger$\quad\ \,}
\end{align}
This completes the proof.
\end{proof}

\begin{rem}\label{final B=final coalg}
  Being an algebra homomorphism and preserving solutions are
  independent concepts: neither of them implies the other one. To see
  this, consider for $HX = X \times X$ an algebra $A=\{a,b\}$ with a
  binary operation $*$ such that $a$ and $b$ are idempotent. We turn
  $A$ into a Bloom algebra by taking $e^\dagger=\mathsf{const}_a$ for
  every $e:X\rightarrow X\times X$. Then there are two homomorphisms
  from the one-point binary algebra (which clearly is corecursive) to
  $A$, yet only one of them is solution preserving. Thus, there exist
  homomorphisms which are not solution preserving. 

  Conversely, there
  exist solution preserving morphisms which are not homomorphisms. To
  see this, let us assume that we have $x*y=a$ for all $x\neq y$ in
  $A$. There are two different structures of Bloom algebras on $A$,
  ($A,\alpha,e\mapsto \mathsf{const}_a$) and ($A,\alpha,e\mapsto
  \mathsf{const}_b$). The map on $A$ which swaps $a$ and $b$ is a
  solution preserving map between the two Bloom algebras, but not a
  homomorphism.

\end{rem}
 \begin{prop}\label{3x}
 An initial Bloom algebra is precisely a terminal coalgebra.
   \end{prop}
More precisely, the statement in Example \ref{ex:binary} generalizes from corecursive algebras to Bloom algebras.
Indeed, the proof in \cite{cuv2} can be used again.
\begin{lem}\label{A is B then B is B}
If $(A,a,\dagger)$ is a Bloom algebra and $h:(A,a)\rightarrow (B,b)$ is a homomorphism of algebras, then there is a unique structure of a Bloom algebra on $(B,b)$ such that $h$ is a solution preserving morphism. We call it, \emph{the Bloom algebra induced by} $h$.
\end{lem}

\begin{proof}
For every coalgebra $e:X\rightarrow HX$ we define

$$e^*\equiv X\stackrel{e^\dagger}\rightarrow A\stackrel{h}\rightarrow B$$
and verify that $e^*$ is a solution of $e$ by the following commutative diagram
 $$
\xymatrix{
X\ar[r]^{e^\dagger}\ar[d]_{e}&A\ar[r]^{h}&B\ar@{<-}`u[l]`[ll]_{e^*}[ll]\\
HX\ar[r]_{He^\dagger}&HA\ar[u]_a\ar[r]_{Hh}&HB\ar[u]_{b}\ar@{<-}`d[l]`[ll]^{He^*}[ll]
}
$$
Functoriality is easily checked too:
let $g:(X,e)\rightarrow (Y,f)$ be a coalgebra homomorphism. Then the following equations hold:
\begin{align}
f^*\cdot g&=h\cdot f^\dagger\cdot g &\text{by the definition of $(-)^*$}\notag\\
&=h\cdot e^\dagger&\text{by functoriality of $\dagger$\qquad}\notag\\
&=e^*&\text{by the definition of $(-)^*$}\notag
\end{align}
Finally, $h$ is clearly solution preserving.

The unicity of the Bloom algebra structure given by $(-)^*$ is clear.
\end{proof}

\begin{rem}
 We are going to characterize the left adjoint of the forgetful functor
\[\begin{array}{cccccc}
U:\mathsf{Alg}_B\,H&\rightarrow &\mathcal A,\qquad
(A,a,\dagger)&\mapsto &A.
\end{array}\]
In other words, we characterize the free Bloom algebras: they are
coproducts $T\oplus FY$ of the terminal coalgebra and free algebras. For that we first attend to the existence of those ingredients.
\end{rem}

\begin{lem}\label{free B-> final coalg}
Let $\mathcal A$ be a complete category. If $H$ has a free Bloom algebra on an object $Y$ with $\mathcal A(Y,HY)\neq \emptyset$, then $H$ has a terminal coalgebra.
\end{lem}
\begin{proof}
The free Bloom algebra $(A,a,\dagger)$ on $Y$, is weakly initial in $\mathsf{Alg}_B\, H$. To see this, choose a morphism $e:Y\rightarrow HY$. For every Bloom algebra $(B,b,\ddag)$ the solution $e^\ddag:Y\rightarrow B$ extends to a homomorphism $h:(A,a,\dagger)\rightarrow (B,b,\ddag)$ of Bloom algebras.

Since $\mathsf{Alg}_B\,H$ is complete by Example \ref{ex 3}(e), we can use Freyd's Adjoint Functor Theorem.\smnote{Should be Freyd's Initial Object Theorem}
The existence of a weakly initial object implies that $\mathsf{Alg}_B\, H$ has an initial object. Now apply Proposition \ref{3x}.
\end{proof}

\takeout{
\begin{exa}\label{set functor}
Every set functor with a free Bloom algebra has a terminal coalgebra. The condition {\bf Set}$(Y, HY)\neq \emptyset$ is here automatically satisfied except when $H$ is constantly $\emptyset$.
\end{exa}
}

\begin{construction}\label{free chain} Free-Algebra Chain.
Recall from \cite{a} that if $\mathcal A$ is cocomplete, we can define a chain constructing the free $H$-algebra on $Y$ as follows:
$$
\xymatrix{
Y\ar[r]^-{\mathsf{inr}}&HY+Y\ar[rr]^-{H\mathsf{inr}+Y}&&H(HY+Y)+Y\ar[r]^{}&\ldots
}
$$
We mean the essentially unique chain $V:\mathsf{Ord}\rightarrow \mathcal A$ with
\begin{align}
V_0&=Y\notag\\
V_{i+1}&=HV_i+Y\notag
\end{align}
and for limit ordinals $i$
\begin{align}
V_j&=\colim_{k<j}V_k\notag
\end{align}
whose connecting morphisms $v_{i,j}:V_i\rightarrow V_j$ are defined by
$$v_{0,1}\equiv \mathsf{inr}:Y\rightarrow HY+Y\ \ \ \text{and} \ \ \ v_{i+1,j+1}\equiv Hv_{i,j}+Y$$
and for limit ordinals $j$
$$(v_{k,j})_{k<j}\ \ {\rm is\ the \ colimit\ cocone.}$$

This chain is called the {\it free-algebra chain}. If it \emph{converges} at some ordinal $\lambda$, that is, if $v_{\lambda,\lambda+1}$ is an isomorphism, then $V_\lambda$ is a free algebra on $Y$. More detailed: this isomorphism turns $V_{\lambda}$ into a coproduct
$$V_\lambda=HV_\lambda+Y$$
and thus $V_\lambda$ is an algebra via the left-hand coproduct injection, and the right-hand one $Y\rightarrow V_\lambda$ yields the universal arrow.
\end{construction}

\begin{cor}\label{c-access}
Every accessible endofunctor of a cocomplete category has free algebras.
\end{cor}
Indeed, if the given functor is $\lambda$-accessible, the free-algebra chain converges at $\lambda$.

\begin{defi}(See \cite{at})
We say that in a given category monomorphisms are {\it constructive} provided that
\begin{enumerate}[label=\({\alph*}]
\item if $m_i:A_i\rightarrow B_i$ are monomorphisms for $i=1,2$ then $m_1+m_2:A_1+A_2\rightarrow B_1+B_2$ is a monomorphism,

\item coproduct injections are monomorphisms, and

\item if $a_i:A_i\rightarrow A$, ($i<\lambda$), is a colimit of an
  $\lambda$-chain and $m:A\rightarrow B$ has all composites $m\cdot
  a_i$ monic, then $m$ is monic.
\end{enumerate}
\end{defi}

\begin{exa}
Sets, posets, graphs, abelian groups have constructive monomorphisms. If $\mathcal A$ has constructive monomorphisms, then all functor categories $\mathcal A^{\mathcal C}$ do. In all locally finitely presentable categories condition (c) holds (see \cite{ar}), but (a) and (b) can fail.
\end{exa}

\begin{prop}\label{free B then free alg}
Let $\mathcal A$ be a cocomplete, wellpowered category with constructive monomorphisms. If $H$ has a free Bloom algebra on $Y$ and preserves monomorphisms, then it also has a free algebra on $Y$.
\end{prop}

\begin{proof}
Let $(A,a,\dagger)$ be a free Bloom algebra and $\eta:Y\rightarrow A$ be the universal arrow.\medskip

\noindent (1) We prove that $A=HA+Y$ with coproduct injection $a$ and $\eta$. The algebra
$$\xymatrix{b\equiv H(HA+Y)\ar[r]^-{H[a,\eta]}& HA\ar[r]^-{\mathsf{inl}}&HA+Y}$$
is a Bloom algebra when we put, as in Lemma \ref{generalized},
$$e^\ddag=\mathsf{inl}\cdot He^\dag\cdot e, \ \ \ {\rm for\ all\ } e:X\rightarrow HX.$$
Indeed, functoriality is easy to verify. Consequently, there exists a unique solution preserving homomorphism
$$h:A\rightarrow HA+Y\ \ \ {\rm with}\ \ \ h\cdot \eta=\mathsf{inr}.$$

\noindent We prove that $h$ is inverse to $[a,\eta]$. For that, observe that $[a,\eta]$ is also a homomorphism
$$
\xymatrix{
H(HA+Y)\ar[d]_{H[a,\eta]}\ar[r]^-{H[a,\eta]}&HA\ar[r]^-{\mathsf{inl}}\ar[rd]^a&HA+Y\ar[d]^{[a,\eta]}\\
HA\ar[rr]_a&&A
}
$$
and also preserves solutions:
$$[a,\eta]\cdot e^\ddag=a\cdot He^\dag\cdot e=e^\dag.$$
The composite with $h$ yields an endomorphism of $(A,a,\dagger)$ such that
$$ ([a,\eta]\cdot h)\cdot \eta=[a,\eta]\cdot \mathsf{inr}=\eta.$$
The universal property thus implies
$$[a,\eta]\cdot h=id_A.$$
Since $h$ is a homomorphism, 
$$h\cdot a=b\cdot Hh=\mathsf{inl}\cdot H([a,\eta]\cdot h)=\mathsf{inl}.$$
We conclude that $h$ is inverse to $[a,\eta]$:
$$h\cdot [a,\eta]=[h\cdot a,h\cdot\eta]=[\mathsf{inl},\mathsf{inr}]=id.$$
Therefore, $A=HA+Y$ with coproduct injections $a$ and $\eta$.\medskip

\noindent(2) We define a cone $m_i:V_i\rightarrow A$ of the free-algebra chain by
$$m_0\equiv \eta:Y\rightarrow A$$
and
$$\xymatrix{ m_{i+1}\equiv  HV_i+Y\ar[rr]^-{Hm_i+id}&& HA+Y\ar[r]^-{[a,\eta]}& A}.$$
More precisely: there is a unique cone with the above properties. The
verification of $m_{i+1}\cdot v_{i,i+1}=m_i$ is an easy transfinite
induction on $i$, and the limit steps then follow automatically from
$V_i=\colim_{k<i}V_k$. Moreover, since monomorphisms are
constructive, all $m_i$'s are monomorphisms: $m_0$ is a coproduct
injection of $A=HA+Y$, if $m_i$ is a monomorphism, then so is
$Hm_i+id$, hence, so is $m_{i+1}$, and the limit steps are
clear. Since $A$ has only a set of subobjects, 
there exist $j>i$ such
that $m_i$ and $m_j$ represent the same subobject, i.\,e., $v_{i,j}$
is an isomorphism. Then also $v_{i+1,j+1}$ is an isomorphism with an
inverse $v'$, say. It follows that $v_{j,j+1}$ is an isomorphism, too:
it is monomorphic since $m_{j+1}\cdot v_{j,j+1}=m_j$, and it is a
split epimorphism since $v_{j,j+1}\cdot v_{i+1,j} \cdot v' =
v_{i+1,j+1} \cdot v' =id$. Therefore, $V_j$ is a free algebra on $Y$.
\end{proof}

\begin{cor}\label{cor:3.16}
If a set functor has a free Bloom algebra on $Y$, it has both a terminal coalgebra $T$ and a free algebra $FY$ on $Y$.
\end{cor}
\begin{proof}
  For the existence of $T$ use Lemma~\ref{free B-> final coalg}: for $\mathcal A = \mathbf{Set}$ we have $\mathcal A(Y, HY) \neq \emptyset$ whenever $HY \neq \emptyset$ or $Y = \emptyset$, and in the remaining case where $Y \neq\emptyset$ and $HY = \emptyset$, $T$ trivially exists since one deduces that $H$ is the constant functor on $\emptyset$; indeed, for an any set $X$ pick some $f: X \to Y$, then $Hf: HX \to HY = \emptyset$ implies $HX = \emptyset$.\smnote{I added the reasoning since this this will help readers and only costs us 3 lines.}

For the existence of $FY$ apply Proposition \ref{free B then free alg} in the case where $H$ preserves monomorphisms. If it does not, redefine it in $\emptyset$ and obtain a set functor $H'$ with $H'X=HX$, for all $X\neq \emptyset$ and $H'\emptyset=\emptyset$. Then $H'$ preserves monomorphisms and, whenever $Y\neq\emptyset$, it has a free Bloom algebra on $Y$ if and only if $H$ does. The case $Y=\emptyset$ is obtained from Proposition \ref{3x}.
\end{proof}

We are going to characterize free Bloom algebras. From Lemma \ref{free B-> final coalg} and Proposition \ref{free B then free alg}
we know that the assumption that $T$ and $FY$ exist is ``natural". Since the terminal coalgebra $\tau:T\rightarrow HT$ has by Lambek's Lemma an invertible structure map $\tau $, we can view it as an algebra. In the next proposition we assume that the coproduct of $T$ and $FY$ exists in $\mathsf{Alg}\, H$. In Section~\ref{sec:4} we will see that this actually follows from the existence of a free Bloom algebra.

\begin{nota}
Coproducts in $\mathsf{Alg}\, H$ are denoted by $(A,a)\oplus(B,b)$.
\end{nota}

\begin{thm}\label{thm:3.16}
Suppose that $H$ has a terminal coalgebra $T$, a free algebra $FY$ on $Y$ and their coproduct $T\oplus FY$. Then the last algebra is the free Bloom algebra on $Y$.
\end{thm}
\begin{rem}\smnote{remark added to address referee's point}
  More precisely, let $\eta: Y \to FY$ be the universal arrow of the free algebra $FY$. Take the unique Bloom algebra structure $\dagger$ on $T \oplus FY$ induced by $\mathsf{inl}:T\rightarrow T\oplus FY$ following Lemma~\ref{A is B then B is B}.  Then this forms the free Bloom algebra $(T \oplus FY, a, \dagger)$ on $Y$ with the universal arrow $\mathsf{inr}\cdot\eta$.  The latter means that for every Bloom algebra $(B,b,\ddagger)$ and every morphism $g: Y \to B$ in $\mathcal A$ there exists a unique homomorphism of Bloom algebras $h: T\oplus FY \to B$ such that $h \cdot \mathsf{inr}\cdot\eta = g$. 
\end{rem}
\begin{proof}
Given a Bloom algebra $(B,b,\ddag)$ and morphism $g:Y\rightarrow B$, we obtain a unique homomorphism $\overline g:(FY,\varphi_Y)\rightarrow (B,b)$ with $g=\overline g\cdot\eta$. We also have a unique solution-preserving homomorphism $f:(T,\tau^{-1},\dagger)\rightarrow (B,b,\ddag)$, see Proposition \ref{3x}. This yields a homomorphism $[f,\overline g]:T\oplus FY\rightarrow B$ which is clearly solution-preserving: recall from Lemma~\ref{A is B then B is B} that solutions in $T\oplus FY$ have the form $\mathsf{inl}\cdot e^\dagger$. Thus, $[f,\overline g]\cdot(\mathsf{inl}\cdot e^\dagger)=f\cdot e^\dagger=e^\ddag.$ And this is the desired morphism since $[f,g]\cdot \mathsf{inr}\cdot\eta=\overline g\cdot\eta=g$.

Conversely, given a solution-preserving homomorphism $h:T\oplus FY\rightarrow B$ with $h\cdot\mathsf{inr}\cdot\eta=g$, then $h=[f,\overline g]$, because $h\cdot\mathsf{inl}:T\rightarrow B$ is clearly solution-preserving, hence $h\cdot\mathsf{inl}=f$. Also $h\cdot\mathsf{inr}$ is a homomorphism from $FY$ with $h\cdot\mathsf{inr}\cdot\eta=g$, thus $h\cdot\mathsf{inr}=\overline g$.
\end{proof}

\begin{cor}\label{free B exists}
Every accessible endofunctor of a locally presentable category has free Bloom algebras. They have the form $T\oplus FY$.
\end{cor}

Indeed, recall that the assumptions mean that there exists an infinite ordinal $\lambda$ such that
\begin{enumerate}[label=\({\alph*}]
\item  $\mathcal A$ is cocomplete and has a set $\mathcal A_\lambda$
  of $\lambda$-presentable objects (that is, objects whose
  hom-functors preserve $\lambda$-filtered colimits) such that every
  object is a $\lambda$-filtered colimit of objects in $\mathcal
  A_\lambda$, and

\item $H$ preserves $\lambda$-filtered colimits.
\end{enumerate}

From this it follows that $\mathsf{Coalg}\, H$ is locally
$\bar\lambda$-presentable, where $\bar\lambda = \max\{\lambda,
\aleph_1\}$ (see \cite{ap}). Thus, this category  has a terminal object, $T$.
We know from Corollary \ref{c-access} that the free algebra $FY$
exists. And $T \oplus FY$ exists since the category of algebras for an
accessible functor on a locally presentable category is itself locally
presentable and therefore comcomplete. 
%

\begin{rem}
For concrete examples of $T\oplus FY$ see Example~\ref{some exams} below.
\end{rem}

\begin{prop}\label{lim of B is B}
Let $\mathcal A$ be a complete category. Then so is $\mathsf{Alg}_B\, H$ and limits of Bloom algebras are formed on the level of $\mathcal A$.
\end{prop}
\begin{proof}
This is completely analogous to the proof of Proposition \ref{lim of core is core}. The verification that the function $e\mapsto \langle e^\dagger_i\rangle$ is functorial is trivial.
\end{proof}

\begin{cor}\label{cor:mono}
  For a complete category $\mathcal A$, the monomorphisms in $\mathsf{Alg}_B\, H$ are precisely those homomorphisms carried by monomorphisms in $\mathcal A$. 
\end{cor}
\begin{proof}
  To see this use that in any category with pullbacks a morhism $m$ is a monomorphism iff its kernel pair consists of two identity morphisms. 
\end{proof}

\section{Free Corecursive Algebras}
\label{sec:4}

For accessible functors $H$ we prove that free corecursive algebras $MY$ exist and, in the case where $H$ preserves monomorphisms, they coincide with the free Bloom algebras $MY=T\oplus FY$. Moreover an iterative construction of these free algebras (closely related to the free algebra chain in \ref{free chain}) is presented.

We first prove that the category of corecursive algebras is strongly epireflective in the category of Bloom algebras. That is, the full embedding is a right adjoint, and the components of the unit of the adjunction are strong epimorphisms.

\begin{prop}\label{cor reflected of B}
For every accessible endofunctor of  a locally presentable category, corecursive algebras form a strongly epireflective subcategory of the category of Bloom algebras. In particular, every Bloom subalgebra of a corecursive algebra in $\mathsf{Alg}_B\, H$ is corecursive.
\end{prop}

\begin{proof}
Let $\lambda$ be an infinite cardinal such that $\mathcal A$ is a
locally $\lambda$-presentable category and $H$ preserves
$\lambda$-filtered colimits. Since $\lambda $ can be chosen
arbitrarily large, we can assume that $\lambda$ is uncountable. The
category $\mathsf{Alg}_B\, H$ is locally $\lambda$-presentable. The
proof is analogous to that of Proposition \ref{dirlim of cor is cor}
(the only difference is that in the proof of the uniqueness of
$e^\dagger$ we simply observe that since $k_t$'s are supposed to be
solution-preserving, we have $e^\dagger =k_t\cdot s$, where $s$ is the
dagger of $e$ in $A_t$). Consequently, $\mathsf{Alg}_B\,H$ is a
complete, well-powered, and cowellpowered category, and it has (strong
epi-mono) factorization of morphisms, see \cite{ar}. The full
subcategory of corecursive algebras is closed under products by
Proposition \ref{lim of core is core}. Thus, the proof will be
completed when we prove that the subcategory of  corecursive algebras
is closed in $\mathsf{Alg}_B\,H$ under subalgebras, then it is
strongly epireflective (see \cite[Theorem~16.8]{ahs}).

Let $m:(A,a,\dagger)\rightarrow (B,b,\ddag)$ be a monomorphism in
$\mathsf{Alg}_B\,H$ with $(B,b,\ddag)$ corecursive. From Corollary~\ref{cor:mono} we have that $m$ is a monomorphism in $\mathcal A$. It is our task to prove that for every coalgebra $e:X\rightarrow HX$ the morphism $e^\dagger$ is the unique solution in $A$. This follows from the fact that $m\cdot e^\dagger=e^\ddag$. Since $e^\ddag$ is unique in $B$ and $m$ is a monomorphism in $\mathcal A$, the proof is concluded.
\end{proof}

\begin{cor}\smnote{added cor. to satisfy referee}
Every accessible endofunctor of a locally presentable category has free corecursive algebras.
\end{cor}
Indeed, since the functors $\mathsf{Alg}_C\,H\hookrightarrow \mathsf{Alg}_B\,H$ and $\mathsf{Alg}_B\,H\rightarrow \mathcal A$ have left adjoints by Corollary \ref{free B exists} and Proposition \ref{cor reflected of B}, their composite has a left adjoint.

\begin{rem}
We believe  that in the generality of the above corollary, the free corecursive algebras are $T\oplus FY$ (as in Corollary \ref{free B exists}). But we can only prove this in the case where $H$ preserves monomorphisms and monomorphisms are constructive. We are going to apply the following transfinite construction of free corecursive algebras closely related to the free algebra construction of \ref{free chain}
\end{rem}

\begin{construction}\label{free core.chain}Free-Corecursive-Algebra Chain.

Let $\mathcal A$ be cocomplete and $H$ have a terminal coalgebra $(T,\tau)$. We define an essentially unique chain $U:\mathsf{Ord} \rightarrow \mathcal A$ by
\begin{align}
U_0&=T\notag\\
U_{i+1}&=HU_i+Y\notag
\end{align}
and for limit ordinals $j$
\begin{align}
U_j&=\colim_{k<j}U_k.\notag
\end{align}
The connecting morphisms $u_{i,j}:U_i\rightarrow U_j$ are defined by
$$u_{0,1}\equiv T\stackrel{\tau}\rightarrow HT\stackrel{\mathsf{inl}}\rightarrow HT+Y$$
$$u_{i+1,j+1}=Hu_{i,j}+id_Y$$
and for limit ordinals $j$
$$(u_{k,j})_{k<j} \text{\ is the colimit cocone}.$$
We say that the chain {\it converges at $\lambda$} if the connecting morphism $u_{\lambda,\lambda+1}$ is an isomorphism, thus $U_\lambda=HU_\lambda+Y$. Then $U_\lambda$ is an algebra (via $\mathsf{inl}$) connected to $Y$ via $\mathsf{inr}$.
\end{construction}

\begin{prop}\label{cor chain converges U=T+FY}
Let $\mathcal A$ be a cocomplete and wellpowered category with constructive mono\-morphisms, and let $H$ preserve monomorphisms and have a terminal coalgebra $T$. If the corecursive chain for $Y$ converges in $\lambda$ steps, then $U_\lambda=T\oplus FY$.
\end{prop}

\begin{rem}
(a) More detailed, we prove that a free algebra $FY$ exists and the algebra $\mathsf{inl}:HU_\lambda\rightarrow U_\lambda$ (obtained from $U_\lambda=HU_\lambda+Y$) is a coproduct of $T$ and $FY$ in $\mathsf{Alg}\, H$.

(b) The proposition is valid for every fixpoint of $H$, not necessarily a terminal coalgebra: given any isomorphism $\tau:T\rightarrow HT$ and forming the corresponding chain with $U_0=T$ and $U_{i+1}=HU_i+Y$, then  whenever it converges, it yields a coproduct of $(T,\tau^{-1})$ and the free algebra on $Y$ in $\mathsf{Alg}\, H$.
\end{rem}
\begin{proof}
(1) A free algebra $FY$ exists. To prove this, we define a natural transformation $(m_i:V_i\rightarrow U_{1+i})_{i\in \mathsf{Ord}}$ from the free-algebra chain to the corecursive chain (delayed on finite ordinals by one step, recall that $1+i=i$ for all infinite ordinals). Put $$m_0=\mathsf{inr}:Y\rightarrow HT+Y$$
and $$m_{i+1}=Hm_i+id_Y:HV_i+Y\rightarrow HU_{1+i}+Y.$$
The first naturality square
$$
\xymatrix{
Y\ar[rr]^-{\mathsf{inr}}\ar[d]_-{\mathsf{inr}}&&HY+Y\ar[d]^-{H\mathsf{inr}+id}\\
HT+Y\ar[rr]_-{H(\mathsf{inl}\cdot\tau)+Y}&&H(HY+Y)+Y
}
$$
 clearly commutes, and the $i$-th square implies the next one easily. Therefore, the limit ordinals $i$ define $m_i$ automatically. Since $H$ preserves monomorphisms, we see by transfinite induction that all $m_i$'s are monic. Consequently, we obtain a transfinite chain of subobjects of $U_\lambda$
 $$\xymatrix{V_i\ar[r]^-{m_i}& U_{1+i}\ar[r]^-{u_{\lambda, 1+i}^{-1}}& U_\lambda,\ \ {\rm for} \ i\geq\lambda}.$$
Since $U_{\lambda}$ has only a set of subobjects, there exist ordinals $\sigma$ and $\rho$ with $\rho > \sigma\geq \lambda$ such that above monomorphisms with $i=\rho$ and $i=\sigma$ represent the same subobject. Similar reasoning as in point~(2) \smnote{clarified reasoning as required by referee}%
of the proof of Proposition~\ref{free B then free alg} then shows that $v_{\rho,\rho+1}$ is an isomorphism using the commutative triangle
$$
\xymatrix{
V_\rho
\ar[rrrr]^{v_{\rho,\rho+1}}
\ar[dr]_{m_\rho}
&&&&
V_{\rho+1}=HV_\rho+Y
\ar[dl]^{m_{\rho+1}}
\\
&U_{1+\rho}\ar[dr]_{u_{\lambda,1+\rho}^{-1}}&&U_{1+\rho+1}\ar[dl]^{u^{-1}_{\lambda,1+\rho+1}}\\
&&U_\lambda
}
$$
We thus proved that the algebra $(V_\rho,\mathsf{inl})$ is free on $Y$ with respect to $\mathsf{inr}:Y\rightarrow V_\rho$. Shortly $FY=V_\rho$.

(2) Analogously, $U_\lambda$ is an algebra with respect to $\mathsf{inl}:HU_\lambda \rightarrow U_\lambda$ (since $u_{\lambda,\lambda+1}$ is invertible). We will prove that this algebra is the coproduct of $T$ and $FY$ with injections $$\overline{\mathsf{inl}}\equiv u_{0,\lambda}:T\rightarrow U_\lambda$$ and $$\xymatrix{\overline{\mathsf{inr}}\equiv FY=V_\rho\ar[r]^-{m_\rho}& U_{1+\rho}\ar[r]^-{u_{\lambda,{1+\rho}}^{-1}}& U_\lambda}$$
in $\mathsf{Alg}\, H$. Indeed, $\overline{\mathsf{inl}}$ is an algebra homomorphism: $\overline{\mathsf{inl}}=u_{1,\lambda}\cdot u_{0,1}$ and we have $u_{1,\lambda+1}=Hu_{0,\lambda}+Y$, therefore the following square commutes:
$$
\xymatrix{
HT\ar[rr]^{\tau^{-1}}\ar[rd]^-{\mathsf{inl}}\ar[dd]_{Hu_{0,\lambda}}&&T\ar[ld]^{u_{0,1}}\ar[dd]^{u_{0,\lambda}}\\
&HT+Y\ar[dr]^{u_{1,\lambda}}\\
HU_\lambda\ar[rr]_-{\mathsf{inl}}&&U_\lambda\simeq HU_\lambda+Y
}
$$

Also $\overline{\mathsf{inr}}$ is a homomorphism: the following diagram

$$
\xymatrix{
HV_\rho\ar[r]^-{\mathsf{inr}}\ar[d]_{Hm_\rho}&HV_\rho+Y\simeq V_\rho\ar[d]^{ Hm_\rho+Y\simeq m_\rho}\\
HU_{1+\rho}\ar[r]^-{\mathsf{inr}}&HU_{1+\rho}+Y\simeq U_{1+\rho}\\
HU_\lambda\ar[u]^{Hu_{\lambda,1+\rho}}\ar[r]^-{\mathsf{inr}}&HU_\lambda+Y\simeq U_\lambda\ar[u]_{Hu_{\lambda,1+\rho}+Y\simeq u_{\lambda,1+\rho}}
}
$$
commutes and yields (by inverting $u_{\lambda,1+\rho}$) the square
$$
\xymatrix{
HFY\ar[r]^{\mathsf{inr}}\ar[d]_{H\overline{\mathsf{inr}}}&FY\ar[d]^{\overline{\mathsf{inr}}}\\
HU_\lambda\ar[r]_-{\mathsf{inr}}&U_\lambda
}
$$

To verify the universal property, let $b:HB\rightarrow B$ be an algebra and $f:T\rightarrow B$ and $g:FY\rightarrow B$ be homomorphisms. We prove that there exists a unique homomorphism $h:U_\lambda\rightarrow B$ with $f=h\cdot\overline{\mathsf{inl}}$ and $g=h\cdot\overline{\mathsf{inr}}$. Put $\hat{g}=g\cdot\mathsf{inr}:Y\rightarrow B$.

{\it Existence}: Define a cocone $h_i:U_i\rightarrow B$ of the chain from
Construction~\ref{free core.chain} by
 $$h_0= f,\quad\text{and}\quad h_{i+1}\equiv \xymatrix@1@+1pc{HU_i+Y\ar[r]^-{Hh_i+Y}& HB+Y\ar[r]^-{[b,\hat g]}& B}$$
The first naturality triangle commutes because $f$ is a homomorphism, thus, $f=b\cdot Hf\cdot\tau$:
$$
\xymatrix{
T\ar[r]^{\tau}\ar[rdd]_f&HT\ar[rr]^-{\mathsf{inl}}\ar[d]^{Hf}&&HT+Y\ar[ld]^{Hf+Y}\\
&HB\ar[r]^-{\mathsf{inl}}\ar[d]^{b}&HB+Y\ar[ld]^-{[b,\hat g]}\\
&B
}
$$
The further verification of the naturality, $h_i=h_{i+1}\cdot u_{i+1}$, is now an easy transfinite induction. We obtain the desired homomorphism $h_\lambda:U_\lambda\rightarrow B$.
Indeed, since $h_{\lambda+1}\cdot\mathsf{inl}=b\cdot Hh_\lambda$ (by the above rule for $h_{i+1}$), the square
$$
\xymatrix{
HU_\lambda\ar[r]^-{\mathsf{inl}}\ar[d]_{Hh_\lambda}&HU_\lambda+Y\simeq U_\lambda\ar[d]^{h_{\lambda+1}}\\
HB\ar[r]_h&B
}
$$
commutes.

The first equation $f=h_\lambda\cdot\overline{\mathsf{inl}}$ follows
from $\overline{\mathsf{inl}}=u_{0,\lambda}$ and $h_\lambda\cdot
u_{0,\lambda}=h_0=f$. For the second one
$g=h_\lambda\cdot\overline{\mathsf{inr}}:FY\rightarrow B$ observe that
both sides are algebra homomorphisms. Thus, it is sufficient to prove
that the universal arrow $\mathsf{inr}:Y\rightarrow V_\rho(\simeq HV_\rho+Y)$ merges them. Recall that $\overline{\mathsf{inr}}=u_{\lambda,1+\rho}^{-1}\cdot m_\rho$, thus we need to prove
$$\hat g=h_\lambda\cdot\overline{\mathsf{inr}}\cdot\mathsf{inr}=h_{1+\rho}\cdot m_\rho\cdot\mathsf{inr}.$$
This follows from $h_{1+\rho+1}=[b,\hat g]\cdot(Hh_{1+\rho}+Y)=[b\cdot Hh_{1+\rho},\hat g]$, thus the outside of the diagram below commutes as desired:
$$
\xymatrix{
Y\ar`d[ddr][rrdd]_(.4){\hat g}\ar[r]^-{\mathsf{inr}}\ar[rd]_{\mathsf{inr}}&HV_\rho+Y\ar[r]^-{\sim}\ar[d]^{Hm_\rho+Y}&V_\rho\ar[d]^{m_\rho}\\
&HU_{1+\rho}+Y\ar[r]^-{\sim}\ar[rd]_{[b\cdot Hh_{1+\rho},\hat g]}&U_{1+\rho}\ar[d]^{h_{1+\rho}}\\
&&B
}
$$

{\it Uniqueness}: Consider an algebra homomorphism
$h:U_\lambda\rightarrow B$ with $h\cdot\overline{\mathsf{inl}}=f$ and
$h\cdot\overline{\mathsf{inr}}=g$, we prove $h\cdot u_{i,\lambda}=h_i$
by transfinite induction on $i\leq\lambda$. The case $i=\lambda$ yields $h=h_\lambda$. The initial step is clear:$$h\cdot u_{0,\rho}=h\cdot\overline{\mathsf{inl}}=f=h_0.$$
Assuming $h\cdot u_{i,\lambda}=h_i$, we are going to prove that the triangle
$$
\xymatrix{
HU_i+Y=U_{i+1}\ar[d]_{{Hu_{i,\lambda}+Y}\simeq u_{i+1,\lambda}}\ar[dr]^(.7)*+{\labelstyle h_{i+1}=[b\cdot Hh_i,\hat g]}\\
HU_\lambda+Y\simeq U_\lambda\ar[r]_-h&B
}
$$
commutes. The left-hand component with domain $HU_i$ commutes because $h$ is a homomorphism, that is, $h\cdot\mathsf{inl}=b\cdot Hh$:
$$
\xymatrix{
HU_i\ar[d]_{{Hu_{i,\lambda}}}\ar[dr]^{Hh_i}\\
HU_\lambda\ar[r]_{Hh}\ar[d]_{\mathsf{inl}}&HB\ar[d]^{b}\\
U_\lambda\ar[r]_h&B
}
$$
The right-hand one with domain $Y$ follows from $g=h\cdot \overline{\mathsf{inr}}=h\cdot u_{\lambda,1+\rho}^{-1}\cdot m_{\rho}$: we have
$$\hat g=g\cdot\mathsf{inr}=h\cdot u^{-1}_{\lambda,1+\rho}\cdot m_{\rho}\cdot\mathsf{inr}=h\cdot u_{i+1,\lambda}\cdot\mathsf{inr},$$
where the last equation follows from the commutative diagram below expressing $u_{i+1,\lambda+1}=Hu_{i,\lambda}+Y$:
$$
\xymatrix{
&HV_\rho+Y\ar[r]^-\sim\ar[d]^{Hm_\rho+Y}&V_\rho\ar[d]^{m_\rho}\\
Y\ar[ur]^-{\mathsf{inr}}\ar[r]^-{\mathsf{inr}}\ar[dr]^-{\mathsf{inr}}\ar[ddr]_-{\mathsf{inr}}&HU_{1+\rho}+Y\ar[d]^{Hu^{-1}_{\lambda,{1+\rho}}+Y}
\ar[r]^-\sim&U_{1+\rho}\ar[d]^{u^{-1}_{\lambda,1+\rho}}\\
&HU_\lambda+Y\ar[r]^-\sim&U_\lambda\\
&HU_i+Y\ar[u]_{Hu_{i,\lambda}+Y}\ar@{=}[r]&U_{i+1}\ar[u]_{u_{i+1,\lambda}}
}
$$
Finally, for a limit ordinal $\alpha$ we easily derive $h\cdot u_{\alpha,\lambda}=h_\alpha$ by extending with the colimit injections $u_{i,\alpha}$, $i\leq \alpha$, and using that they are jointly epic: $$h\cdot u_{\alpha,\lambda}\cdot u_{i,\alpha}=h\cdot u_{i,\lambda}=h_i=h_\alpha\cdot u_{i,\alpha}.$$
\end{proof}

\begin{thm}\label{free cor exists}
Let $\mathcal A$ be a locally presentable category with constructive monomorphisms. Every accessible endofunctor preserving monomorphisms has free corecursive algebras $MY=T\oplus FY$.
\end{thm}

\begin{proof}
From Corollary \ref{free B exists} we know that $T\oplus FY$ is a free Bloom algebra, thus, it is sufficient to prove that this algebra is corecursive. For that, we use Proposition \ref{cor reflected of B} and find a corecursive algebra such that $T\oplus FY$ is its subalgebra; this will finish the proof.

The endofunctor $H(-)+Y$ is also accessible. Thus, it also has a terminal
coalgebra (see the proof of Corollary~\ref{free B exists}). We denote it by
$TY$. The components of the inverse of its coalgebra structure
$TY\stackrel{\sim}\rightarrow H(TY)+Y$ are denoted by
$\tau_Y:H(TY)\rightarrow TY$ and $\eta_Y:Y\rightarrow TY$, respectively. As proved in \cite{m_cia} the algebra $TY$ is a cia for $H$, cf. Remark \ref{cia}. We are going to prove that $T\oplus FY$ is a subalgebra of this $H$-algebra $TY$.

Since $H$ is accessible and preserves monomorphisms, the terminal chain
of $H$ converges (and yields a terminal coalgebra $T$), as proved in \cite{at2}. That is, if we define an chain $W:\mathsf{Ord}^{op}\rightarrow \mathcal A$ on objects by
$$W_0=1\ \ \ {\rm and}\ \ \ W_{i+1}=HW_i$$
with $W_i=\lim_{k<i}W_k$ for limit ordinals  (and on morphisms by
$w_{1,0}:H1\rightarrow 1$ unique, $w_{i+1,j+1}=Hw_{i,j}$ and
$(w_{i,k})_{k<i}$ forming a limit cone), then there exists an ordinal
$\lambda$ such that $w_{\lambda+1,\lambda}:HW_\lambda\rightarrow
W_\lambda$ is invertible. We then get 
$$T=W_\lambda\ \ \ {\rm and}\ \ \ \tau=w_{\lambda+1,\lambda}^{-1}.$$
We can choose $\lambda$ arbitrarily large, thus we can assume that $\lambda$ is a cardinal such that $H$ is $\lambda$-accessible.

Since monomorphisms are constructive, $H(-)+Y$ also preserves monomorphisms and is also $\lambda$-accessible. Denote by $\overline W:\mathsf{Ord}^{op}\rightarrow \mathcal A$ its terminal chain. Then this chain converges and yields a terminal coalgebra $TY$. Since we again can choose an arbitrary large ordinal for the convergence of $\overline W$, we can assume that this is the above cardinal $\lambda$, thus, $TY=\overline W_\lambda$ and $[\tau_Y,\eta_Y]=\overline w_{\lambda+1,\lambda}:H(TY)+Y\rightarrow TY$. We conclude that $T$ is a (canonical) subalgebra of $TY$: define a natural transformation $m_i:W_i\rightarrow \overline W_i$ by $m_0=id_1$ and $m_{i+1}\equiv HW_i\stackrel{Hm_i}\longrightarrow H\overline W_i\stackrel{\mathsf{inl}}\longrightarrow H\overline W_i+Y$. The constructivity of monomorphisms implies that $\mathsf{inl}$ is monic, thus, we see by easy transfinite induction that $m_i$'s are monomorphisms for all $i\in \mathsf{Ord}$. And $m_\lambda:T\rightarrow TY$ is a coalgebra homomorphism because the $\lambda$-th naturality square of ($m_i)$ yields $m_\lambda\cdot\tau=\tau_Y\cdot Hm_\lambda$:
$$
\xymatrix{
HT\ar[d]_{\mathsf{inl}\cdot Hm_\lambda}&&T\ar[ll]_-{\tau}\ar[d]^{m_\lambda}\\
H(TY)+Y&&TY\ar[ll]^-{[\tau_Y,\eta_Y]^{-1}}
}
$$

We now define a cocone $p_i: U_i\rightarrow TY$, for $i\leq\lambda$, of
the chain from Construction~\ref{free core.chain} by
$$p_0=m_\lambda:T\rightarrow TY$$
and
$$\xymatrix{p_{i+1}\equiv HU_i+Y\ar[rr]^-{Hp_i+Y}&& H(TY)+Y\ar[rr]^-{[\tau_Y,\eta_Y]}&&TY}.$$
We need to verify compatibility, $p_i=p_{i+1}\cdot u_{i,i+1}$, from which the limit steps follow automatically. For $i=0$, use that $m_\lambda$ is a coalgebra homomorphism:
$$
\xymatrix{
T\ar[r]^\tau\ar[rdd]_{m_\lambda}&HT\ar[rr]^{\mathsf{inl}}\ar[d]^-{Hm_\lambda}&&HT+Y\ar[ld]^(.4)*+{\labelstyle Hm_\lambda+Y}\ar@{<-}`u[l]`[lll]_-{u_{0,1}}[lll]\\
&HTY\ar[r]^-{\mathsf{inl}}\ar[d]^-{\tau_Y}&H(TY)+Y\ar[ld]^(.4)*+{\labelstyle [\tau_Y,\eta_Y]}\\
&TY
}
$$
The isolated  step is easy because if the compatibility holds for $i$, then the diagram
$$
\xymatrix{
HU_i+Y\ar[rr]^{Hu_{i,i+1}+Y}\ar[rd]^{Hp_i+Y}\ar[rdd]_{p_{i+1}}&&HU_{i+1}+Y\ar[ld]_{Hp_{i+1}+Y}\ar[ddl]^{p_{i+2}}\\
&H(TY)+Y\ar[d]|(.45){[\tau_Y,\eta_Y]}\\
&TY
}
$$
commutes. It is obvious (by transfinite induction) that all $p_i$'s
are monomorphisms. It remains to verify that
$p_\lambda:U_\lambda\rightarrow TY$ is an algebra
homomorphism. Indeed, since $H$ preserves $\lambda$-filtered colimits,
the corecursive chain $U_i$ converges after $\lambda$ steps, so that $U_\lambda = T\oplus FY$, by Proposition \ref{cor chain converges U=T+FY}. We have following commutative square:
$$
\xymatrix@C+1pc{
HU_\lambda\ar[r]^-{\mathsf{inr}}\ar[d]_{Hp_\lambda}&HU_\lambda+Y\simeq U_\lambda\ar[d]^{{p_{\lambda+1}}\simeq p_\lambda}\\
H(TY)\ar[r]_-{\tau_Y(=\mathsf{inr})}&H(TY)+Y\simeq TY 
}
$$

\vspace*{-20pt}
\end{proof}

\begin{exa}\label{some exams}
Free corecursive algebras $MY$ that are obtained as $U_\omega$.
\begin{enumerate}
\item For $H=Id$ we have
$$MY=U_\omega= 1+Y+Y+Y+\ldots$$
Indeed, the terminal object $1$ is $T$, and
\begin{align}
U_1&=T+Y&
U_2&=T+Y+Y&
\cdots\notag
\end{align}
with colimit $U_\omega=1+Y+Y+Y+\cdots$.

\item More generally, let $H:\mathcal A\rightarrow \mathcal A$ preserve countable coproducts and have a terminal coalgebra $T$. Then $$MY=U_\omega=T+Y+HY+H^2Y+\cdots$$

\item For the endofunctor $HX=X\times X$ of {\bf Set} (of binary algebras) we have the free corecursive algebra
$$MY=\text{all binary trees with finitely many leaves, all of which are labelled in $Y$.}$$
Indeed, the corecursive chain $U$ yields
$$U_0=1$$
which we represent by the complete binary tree $t$. Then
$$U_1=1\times 1+Y \simeq 1+Y$$
is represented by $t$ and all singleton trees labelled in $Y$, and
$$U_2=(1+Y)\times (1+Y)+Y$$
is represented by all binary trees with leaves of depth at most $1$ labelled in $Y$.
We conclude that
$$U_n=\text{all binary trees with no leaves of depth $\geq n$ and leaves labelled in $Y$}.$$
Consequently, the corecursive chain yields
$$T\oplus FY=\colim_{n<\omega}U_n=\bigcup_{n<\omega}U_n$$
which is the above set of trees. Observe that this description of $MY$
corresponds well with the fact that $MY$ is a free binary algebra with
an additional idempotent (see Example \ref{ex:binary}): the unique
idempotent of $MY$ is the complete binary tree. And $MY$ is generated
by this tree and all finite trees in $FY$.

\item More generally, let $\Sigma=(\Sigma_k)_{k<\omega}$ be a signature. Then $\Sigma$-algebras are precisely the algebras for the polynomial
endofunctor
$$H_\Sigma X=\Sigma_0+\Sigma_1\times X+\Sigma_2\times X^2+\cdots$$
Recall that the terminal coalgebra is the coalgebra of all {\it $\Sigma $-trees}, that is, trees labelled in $\Sigma$ so that every node with a label of arity $n$ has precisely $n$ children. And $FY$ is the algebra of all finite $(\Sigma+Y)$-trees, where members of $Y$ are considered to have arity $0$. Then $U_n$ is the set of all $(\Sigma+Y)$-trees with no leaf of depth greater than $n$ having a label from $Y$. (That is, all leaves on level $n$ or more are labelled by a nullary symbol in $\Sigma_0$.) Consequently the free corecursive algebra $\colim_{n<\omega}U_n=\bigcup_{n<\omega}U_n$ is
\[
MY=T\oplus FY= \parbox[t]{9cm}{all $(\Sigma+Y)$-trees in which finitely many leaves
are labelled in $Y$ and other leaves labelled in $\Sigma_0$.}
\]

\item For the finite power set functor $\mathcal P_f$, J. Worrell
\cite{w} described the terminal coalgebra $T$ as the coalgebra of all finitely branching,
non-ordered, strongly extensional  trees. Recall that a (non-ordered)
tree is called {\it strongly extensional} if for every node, the subtrees rooted at distinct children of the given node are not tree bisimilar. The free corecursive algebra is
\[
T\oplus FY =
\parbox[t]{10cm}{all finitely branching, strongly extensional trees with finitely
 \quad many leaves labelled in $Y$ (and other leaves unlabelled).
}
\]
Indeed, this is analogous to (3) since
\[
\begin{array}{rcccp{9cm}}
  U_1 & \simeq & T+Y & = & all the trees in $T$ and all singleton
  trees labelled in $Y$,\\
  U_2 & = & \mathcal P_fU_1+Y & = &
  all finitely branching, strongly extensional trees with
  some leaves of depth less than $2$ labelled in $Y$,
  \\
  \multicolumn{2}{l}{\textrm{etc.}}
\end{array}
\]
Thus $T\oplus FY=\bigcup_{n<\omega}U_n$ is the above set of trees. The
corecursive chain converges in $\omega$ steps  because $\mathcal P_f$
preserves $\omega$-colimits.
\end{enumerate}
\end{exa}
\begin{exa}
Now we illustrate the need of $U_i$ for infinite ordinals $i$. The
functor $H:\mathbf{ Set}\rightarrow \mathbf{ Set}$ with
$HX=X^{\mathbb{N}}$ has the following free corecursive algebras, where
$t_0$ is the is the complete countably branching tree:
\[
MY =
\parbox[t]{10cm}{all countably branching trees with leaves
labelled in $Y$ such that every infinite path contains a node whose
subtree is $t_0$.}
\]
To see this, represent $U_0=1$ by the single tree $t_0$ and obtain
\[\begin{array}{rcccp{10cm}}
 U_1&\simeq& T+Y&=& all singleton trees labelled in $Y$, plus $t_0$ \\
U_2&\simeq& (T+Y)^\mathbb{N}+Y&=& all countably branching trees with
leaves at levels at most $1$, all of which are labelled in $Y$\\
&&&\vdots\notag\\
U_\omega &\simeq &T+Y&=&
all countably branching trees which, for some $n< \omega$, have all leaves at level at most $n$, and they are labelled in $Y$.
\end{array}\]
 But $U_\omega$ is not an algebra because it does not contain the following tree
 $$
 
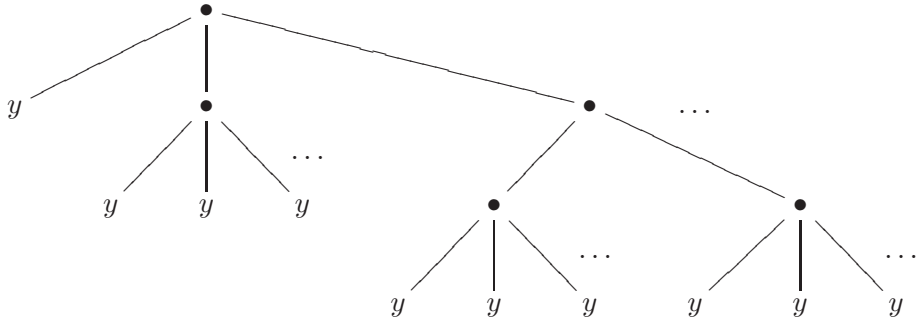
\xymatrix{
 &&\bullet\ar@{-}[lld]\ar@{-}[d]\ar@{-}[rrrrd]\\
 y&&\bullet\ar@{-}[ld]\ar@{-}[d]\ar@{-}[rd]^(.7)*+{\ldots}&&&&\bullet\ar@{-}[ld]\ar@{-}[rrd]&\ldots\\
 &y&y&y&&\bullet\ar@{-}[ld]\ar@{-}[d]\ar@{-}[rd]^(.7)*+{\ldots}&&&\bullet\ar@{-}[ld]\ar@{-}[d]\ar@{-}[rd]^(.7)*+{\ldots}\\
 &&&&y&y&y&y&y&y
}
 $$
 This tree is in $U_{\omega+1}=U_\omega^\mathbb{N}+Y$.

 We can identify, for every ordinal $i$, the set $U_i$ with the set of all countably branching trees of type $i$ where type $0$ means the tree is $t_0$, type $i+1$ means that all maximal subtrees have type $i$, and for limit ordinals $j$, type $j$ means type $i$ for some $i<j$. Then
 $$U_{\omega_1}=\text {all countably branching trees of countable type.}$$
 It is easy to verify that a tree has countable type if and only if on every infinite path there exists a node whose subtree is $t_0$. Thus, $U_{\omega_1}=T\oplus FY$ is the above free corecursive algebra.
 \end{exa}

For the proof of Theorem~\ref{equivalence 1}, the main result of this section, we are going to use the following technical lemma concerning fixpoints and coproducts of algebras.\smnote{explanation added due to referee comment}

\begin{lem}\label{fixed point lemma}
Let $a:HA\stackrel{\sim}\to A$ be a fixpoint and $B=HB+Y$. If the algebras $(A,a)$ and $(B,\mathsf{inl})$ have a coproduct $(C,c)=(A,a)\oplus (B,\mathsf{inl})$ with injections $i$ and $j$ in $\mathsf{Alg}\,H$, then we have a coproduct $C=HC+Y$ in $\mathcal A$ with injections
$$c:HC\to C\qquad \text{and}\qquad  d\equiv Y\stackrel{\mathsf{inr}}\to B\stackrel{j}\to C$$
\end{lem}

\begin{proof}
For the algebra $\overline C=HC+Y$ with structure
$$\xymatrix{
\overline c\equiv H(HC+Y)\ar[r]^-{H[c,d]}&HC\ar[r]^-{\mathsf{inl}}&HC+Y
}$$
we observe that $a^{-1}\cdot Hi\cdot \mathsf{inl}:A\to \overline C$ and $Hj+id:B\to \overline C$  are homomorphisms.
Indeed, for the former we have the commutative diagram
$$\xymatrix@C+1pc{
HA\ar[rr]^-a\ar[d]_-{Ha^{-1}}\ar@{=}[dr]&&A\ar[d]^-{a^{-1}}\\
HHA\ar[r]^-{Ha}\ar[d]_-{HHi}&HA\ar[dd]_-{Hi}&HA\ar[d]^-{Hi}\\
HHC\ar[d]_-{H\mathsf{inl}}\ar[dr]^-{Hc}&&HC\ar[d]^-{\mathsf{inl}}\\
H(HC+Y)\ar[r]_-{H[c,d]}&HC\ar[r]_-{\mathsf{inl}}&\overline{C}
}$$
whose middle left-hand part commutes because $i$ is a homomorphism. For the latter one consider the diagram
$$\xymatrix@C+1pc{
HB\ar[rr]^{\mathsf{inl}}\ar[d]_-{H(Hj+id)}\ar[rd]^-{Hj}&&B=HB+Y\ar[d]^-{Hj+Y}\\
H(HC+Y)\ar[r]_-{H[c,d]}&HC\ar[r]_-{\mathsf{inl}}&HC+Y
}$$
The left-hand triangle commutes because $d=j\cdot\mathsf{inr}$ and because $j$ is a homomorphism, that is, $c\cdot Hj=j\cdot \mathsf{inl}$. Therefore, we have a unique homomorphism $h:C\to \overline C$ with
\begin{equation}\label{eq 4.1}
h\cdot i=\mathsf{inl}\cdot Hi\cdot a^{-1}
\end{equation}
and
\begin{equation}\label{eq 4.2}
h\cdot j= Hj+id.
\end{equation}
This is the inverse of $[c,d]:\overline C\to C$. Indeed, we prove $[c,d]\cdot h=id$ by using the universal property of coproducts. Firstly, $[c,d]$ is a homomorphism:
$$\xymatrix{
H(HC+Y)\ar[r]^-{H[c,d]}\ar[d]_-{H[c,d]}&HC\ar[r]^-{\mathsf{inl}}\ar[dr]^-c&HC+Y\ar[d]^-{[c,d]}\\
HC\ar[rr]_-c&&C
}$$
Therefore, $[c,d]\cdot h$ is an endomorphism of $(C,c)$. And we have both
\begin{align}
[c,d]\cdot h\cdot i&=[c,d]\cdot \mathsf{inl}\cdot Hi\cdot a^{-1}&\text{by (\ref{eq 4.1})\qquad\qquad\qquad}\notag\\
&=c\cdot Hi\cdot a^{-1}\notag\\
&=i\cdot a\cdot a^{-1}&\text{$i$ is a homomorphism}\notag\\
&=i\notag
\end{align}
and
\begin{align}
[c,d]\cdot h\cdot j &=[c,d]\cdot(Hj+id)&\text{by (\ref{eq 4.2})}\notag\\
&=[c\cdot Hj,d]\notag\\
&=[j\cdot\mathsf{inl},j\cdot\mathsf{inr}]&\text{by (\ref{eq 4.1})}\notag\\
&=j.\notag
\end{align}
The remaining identity $h\cdot [c,d]=id$ translates into $h\cdot c=\mathsf{inl}$ and $h\cdot d=\mathsf{inr}$. The former equality follows from $[c,d]\cdot h=id$ and the fact that $h$ is a homomorphism
$$\xymatrix{
HC\ar[rr]^-c\ar[d]_-{Hh}\ar@{=}[dr]&&C\ar[d]^-h\\
H\overline C\ar[r]_-{H[c,d]}&HC\ar[r]_-{\mathsf{inl}}&\overline C
}$$
the latter one follows from (\ref{eq 4.1}) and (\ref{eq 4.2}) since $h\cdot d=h\cdot j\cdot\mathsf{inr}=(Hj+id)\cdot\mathsf{inr}=\mathsf{inr}$.
\end{proof}

\begin{rem}\label{remark pre fixed point}\hfill
\begin{enumerate}[label=\({\alph*}]
\item A {\it pre-fixpoint} of a functor $H$ is an object $A$ such that $HA$ is a subobject of $A$.

\item A fixpoint, i.\;e.~ an object $A\simeq HA$, can be considered as
  an algebra or a coalgebra for $H$. When we speak about {\it
    corecursive fixpoints}, we mean fixpoints
  $\alpha:HA\stackrel{\sim}\rightarrow A$ that are corecursive
  algebras.
\end{enumerate}
\end{rem}

\begin{thm}\label{equivalence 1}
For every set functor, the following statements are equivalent:
\begin{enumerate}[label=(\roman*)]
\item free corecursive algebras exist,
\item free algebras and a terminal coalgebra exist, and
\item arbitrarily large pre-fixpoints and a corecursive fixpoint
  exist.
\end{enumerate}
They imply that the free corecursive algebra on $Y$ is $T\oplus FY$.
\end{thm}
\begin{proof}
  We can assume, without loss of generality, that $H$ preserves
  monomorphisms: see the argument for Corollary~\ref{cor:3.16}.  The
  equivalence of (ii) and (iii) follows from the fact that (a)
  arbitrarily large pre-fixpoints are necessary and sufficient for the
  existence of free algebras, see Theorem II.4 in \cite{takr} and (b)
  every corecursive fixpoint is an initial corecursive algebra, thus,
  a terminal coalgebra, see Proposition 7 in \cite{cuv2}.

Let us prove (i) $\Leftrightarrow$ (ii).

\smallskip\noindent
${\rm (i)\Rightarrow (ii)}$: This follows from Lemma \ref{free B-> final coalg} and Proposition \ref{free B then free alg} which were formulated for Bloom algebras, but hold for corecursive algebras as well. The condition ${\bf Set}(Y,HY) \neq \emptyset$ is automatic for all set functors except the constant functor on $\emptyset$, and (ii) is trivial in that case.

\smallskip
\noindent
${\rm (ii)\Rightarrow(i)}$: Precisely as in the proof of Theorem \ref{free cor exists} we need to prove that $T\oplus FY$ exists and is a subalgebra of a corecursive algebra. The existence is clear, since $\mathsf{Alg}\,H$ is cocomplete. Indeed, the existence of free algebras implies that $H$ generates a free monad $\mathbb{F}$, and then $\mathsf{Alg}\,H$ is equivalent to the category ${\bf Set}^{\mathbb F}$ of Eilenberg-Moore algebras for $\mathbb F$, see \cite{b}. And monadic algebras over {\bf Set} always form cocomplete categories.
We prove that the chain from Construction~\ref{free core.chain}
converges, thus, from Proposition \ref{cor chain converges U=T+FY} we
derive $T\oplus FY=U_\gamma$ for some ordinal $\gamma$. We apply Lemma
\ref{fixed point lemma} to the coproduct
$(C,c)=(T,\tau^{-1})\oplus(FY,\phi_Y)$ with injections
$\overline{\mathsf{inl}}, \overline {\mathsf{inr}}$: put
$d=\mathsf{inr}\cdot \eta_Y:Y\to C$, then $C=HC+Y$ with injections
$c,d$. Let us define a cocone $m_i:U_i\to C$ of the chain $U_i$ by $m_0=\overline{\mathsf{inl}}:T\to C$ and
$$\xymatrix{
m_{i+1}\equiv HU_i+Y\ar[rr]^-{Hm_i+id}&&HC+Y\ar[r]^-{[c,d]}&C.
}$$
We verify that $m_i$ form a cocone: the first triangle commutes since $\overline {\mathsf{inl}}$ is a homomorphism:
$$
\xymatrix{
  T
  \ar[r]^-{\tau}
  \ar[ddr]_-{m_0=\overline{\mathsf{inl}}}
  &
  HT
  \ar[rr]^-{\mathsf{inl}}
  \ar[d]_-{H\overline{\mathsf{inl}}}
  &&
  HT+Y
  \ar[ld]^-{H\overline{\mathsf{inl}}+Y}
  \ar `d[dd] [ddll]^-{m_1}
  \\
  &
  HC\ar[r]^-{\mathsf{inl}}\ar[d]_-c&HC+Y\ar[ld]^-{[c,d]}
  \\
  &
  C && }
$$
If the $i$-th triangle commutes, it is immediately seen that
the next one commutes, too. Furthermore, each $m_i$ is a monomorphism. Indeed,
$m_0=\overline{\mathsf{inl}}$ is a split monomorphism: choose any
morphism $f_0:Y\to T$ in $\mathbf {Set}$ and extend it to a homomorphism
$f:(FY,\phi_Y)\to (T,\tau^{-1})$, then $[id_T,f]$ splits
$\overline{\mathsf{inl}}$. If $m_i$ is a monomorphism, then so is
$m_{i+1}$ because $[c,d]$ is an isomorphism. And limit steps are clear.

Since $C$ has only a set of subobjects,\smnote{referee asked to show details of proof of $u_{\gamma,\gamma+1}$ iso} we know that there exist $\delta > \gamma$ such that $m_\gamma$ and $m_\delta$ represent the same subobject, i.\,e., $u_{\gamma,\delta}: U_\gamma \to U_\delta$ is an isomorphism. Then also $u_{\gamma+1,\delta+1}: U_{\gamma+1} \to U_{\delta+1}$ is clearly an isomorphism with an inverse $u': U_{\delta+1} \to U_{\gamma+1}$, say. It follows that $u_{\delta, \delta+1}$ is an isomorphism, too: it is a monomorphism since $m_{\delta+1} \cdot u_{\delta,\delta+1} = m_\delta$ and a split epimorphism since
$
u_{\delta, \delta+1} \cdot u_{\gamma+1,\delta} \cdot u' = id. 
$
Thus we conclude that $T\oplus FY=U_\delta$.

The rest of the proof is analogous to the proof of Theorem
\ref{free cor exists}, except that we are not using the terminal
coalgebra $TY$ of $H(-)+Y$ (which need not exist), but one of its
``approximations" given by the members $\overline W_i$ of the terminal
chain. I more detail, for every ordinal $i$ the algebra $\overline
w_{i+1,i}:H\overline W_i+Y\rightarrow \overline W_i$ for $H(-)+Y$ is
corecursive (apply Example \ref{exam of core}(5) for $H(-)+Y$). It is
easy to derive that $\overline W_i$ is corecursive also when
considered as an algebra for $H$:
$$\xymatrix{
H\overline W_i\ar[r]^-{\mathsf{inr}} &H\overline W_i+Y\ar[r]^-{\overline w_{i+1,i}}& \overline{W_i}
}$$
Now the existence of the terminal coalgebra $T$ for $H$ implies, as proved in \cite{ak2}, that the terminal chain $W$ for $H$ converges. Let $\lambda$ be an ordinal with $T=W_\lambda$ and $\tau=w^{-1}_{\lambda+1,\lambda}$. As in the proof of Theorem \ref{free cor exists} we define a natural transformation $m_i:W_i\rightarrow \overline W_i$ by $$m_0=id_1\ \ \ \text{ and} \ \ \ m_{i+1}=\mathsf{inl}\cdot Hm_i$$
and after $\lambda$ steps obtain a homomorphism
$m_\lambda:T\rightarrow \overline W_\lambda$ of algebras for $H$ which is monic. And then we obtain a cocone of the corecursive chain $p_i:U_i\rightarrow \overline W_\lambda$ by
$$\xymatrix{ p_0 = m_\lambda &{\rm and}& p_{i+1}\equiv HU_i+Y\ar[rr]^-{Hp_i+id}&&H\overline W_\lambda+Y\ar[r]^-{\overline w_{\lambda+1,\lambda}}&\overline W_\lambda.}$$
Analogously to the proof of Theorem \ref{free cor exists}, if the corecursive chain converges after $\gamma$ steps, then
$$p_\gamma:T\oplus FY=U_\gamma\rightarrow \overline W_\lambda$$
is the desired monomorphism in $\mathsf{Alg}\,H$.
\end{proof}

\begin{nota}\label{4.13}
Let $H$ have free corecursive algebras $MY$. Denote by $$\delta_Y:HMY\rightarrow MY$$ the algebra structure and by $$\eta_Y:Y\rightarrow MY$$ the universal map. Then we obtain a unique homomorphism $$\mu_Y:(MMY,\delta_{MY})\rightarrow (MY,\delta_Y)$$ with $\mu_Y\cdot\eta_{MY}=id$:
\begin{equation}\label{mu is homo}\vcenter{
\xymatrix{
HMMY\ar[r]^-{\delta_{MY}}\ar[d]_{H\mu_Y}&MMY\ar[d]^{\mu_Y}&MY\ar[l]_-{\eta_Y}\ar[ld]^{id_{MY}}\\
HMY\ar[r]_-{\delta_Y}&MY
}}
\end{equation}
The triple $$\mathbb{M}=(M,\mu,\eta)$$ is the monad generated by the adjoint situation  
$$
\xymatrix@1{
\mathsf{Alg}_C\,H \ar@<-5pt>[r] 
\ar@{}[r]|-\perp
& 
\mathcal A
\ar@<-5pt>[l]
}.
$$
\end{nota}

\begin{exa}\label{4.14}
We have $$MY={\mathbb N}\times Y+1$$
for the identity functor on {\bf Set}, see Example \ref{some exams}(1). And for $HX=X\times X$ we have 
$$MY=\text{binary trees with finitely many leaves, all of which are labelled in $Y$,}$$
 see Example \ref{some exams}(3). The functor $HX=\coprod_{n<\omega} X^n$ generates the monad
$$MY=\text{finitely branching trees with finitely many leaves labelled in $Y$}$$ cf. Example \ref{some exams}(4) for $\Sigma$ with one $n$-ary operation for every $n<\omega$.
\end{exa}

\begin{rem}\label{4.15}
Since $\mu_Y$ is, by definition, a homomorphism:
$$\mu_Y\cdot\delta_{MY}=\delta_Y\cdot H\mu_Y$$
the unit law $\mu_Y\cdot\eta_{MY}=id$ yields
$$\delta_Y=\mu_Y\cdot\delta_{MY}\cdot H\eta_{MY}.$$
It easy to prove that the $\delta_Y$ are the components of a natural transformation $\delta:HM\rightarrow M$. 
\end{rem}

\begin{thm}\label{the bloom alg=eilenberg moor off M_H}
Let $\mathcal A$ be a locally presentable category with constructive
monomorphisms, and let $H: \mathcal A \to \mathcal A$ be an accessible
endofunctor preserving monomorphisms. Then Bloom algebras are
precisely the Eilenberg-Moore algebras for $\mathbb M$, i.e., the
category $\mathsf{Alg}_B\,H$ is isomorphic to ${\mathcal A}^{\mathbb  M}$. 
\end{thm}

\begin{proof}
The forgetful functors $U_B:\mathsf{Alg}_B\,\rightarrow \mathcal A$ and $U_C:\mathsf{Alg}_C\,H\rightarrow \mathcal A$ are both right adjoints, and their left adjoints are, in  both cases, defined by $Y\mapsto T\oplus FY$, see Theorem \ref{thm:3.16} and Proposition \ref{cor chain converges U=T+FY}. Therefore the monads generated by $U_B$ and $U_C$ are the same, namely the monad $\mathbb M$. We are going to prove that the comparison functor from $\mathsf{Alg}_B\,H$ to ${\mathcal A}^{\mathbb M}$ is an isomorphism. By Beck's Theorem, see \cite[4.4.4]{bor}, all we need to prove is that $U_B$ creates coequalizers of $U_B$-split pairs. That is, we need to prove that given a parallel pair of solution preserving homomorphisms
$$f,g:(A,a,\dagger)\rightarrow (B,b,\ddagger)$$
and given morphisms $k,s,t$ in $\mathcal A$ as follows:
$$k:B\rightarrow C\ \ \ \text{with}\ \ \ k\cdot f=k\cdot g$$
$$s:C\rightarrow B\ \ \ \text{with}\ \ \ k\cdot s=id_A$$ and
$$t:B\rightarrow A\ \ \ \text{with}\ \ \ s\cdot k=f\cdot t\ \ \ \text{and}\ \ \ id_B=g\cdot t$$
there exists a unique structure $(C,c,*)$ of a Bloom algebra such that $k$ is a solution preserving homomorphism; moreover, $k$ is then a coequalizer in $\mathsf{Alg}_B\,H$. Since $H$ is accessible, free $H$-algebras exist, see Corollary \ref{c-access}. The corresponding monad  is the free monad on $H$ and its Eilenberg-Moore algebras are precisely the $H$-algebras, see \cite{b}. Consequently, by Beck's Theorem there exists a unique structure $c:HC\rightarrow C$ of an algebra for which $k:(B,b)\rightarrow (C,c)$ is a coequalizer in $\mathsf{Alg}\,H$.
 And by Lemma \ref{A is B then B is B}  there exists a unique
 structure $(C,c,*)$ of a Bloom algebra for which $k$ is a solution
 preserving algebra homomorphism. It remains to verify that $k$ is a coequalizer in $\mathsf{Alg}_B\,H$. To this end, let
$$h:(B,b,\ddagger)\rightarrow (D,d,+)$$ be a solution preserving
algebra homomorphism with $h\cdot f=h\cdot g$. There exists a unique homomorphism $h':(C,c)\rightarrow (D,d)$ with $h=h'\cdot k$. And $h'$ preserves solutions (that is for every $e:X\rightarrow HX$ we have $h'\cdot e^*=e^+$) because both $k$ and $h$ do:
$$h'\cdot e^*=h'\cdot k\cdot e^\ddagger=h\cdot e^\ddagger=e^+.$$

\vspace*{-17pt}
\end{proof}

\begin{rem}
In the case $\mathcal A=\mathbf{Set}$ we do not need to assume accessibility: whenever an endofunctor $H$ has free corecursive algebras, then Bloom algebras are precisely the algebras for the monad $\mathbb M$. The proof just uses Theorem \ref{equivalence 1} in lieu of \ref{cor chain converges U=T+FY}.
\end{rem}

\begin{exa}
A non-accessible set functor having free corecursive algebras. Let $\mathbb A$ and $\mathbb B$ be proper classes of infinite cardinals such that for every cardinal $\alpha \in \mathbb A$ the interval $[\alpha,2^\alpha]$ is disjoint with $\mathbb B$. The functor $H$ assigns to every set $X$ the Set
$$HX=\{M\subseteq X\ |\ \text{card $M\in \mathbb B$ or $M=\emptyset$}\}$$
and to every function $f:X\to Y$ the following function
$$Hf(M)=\left \{ \begin{array}{ll}
f[M]& \text{if $f$ is monic when restricted to $M$}\\
\emptyset& \text{otherwise}
\end{array} \right. $$

Then $H$ has the terminal coalgebra $1$ (since $H1=1$). It also has free
algebras. Indeed, this follows from the fact that $H$ has arbitrary large pre-fixpoints (see the
proof of Theorem~\ref{equivalence 1}): for every cardinal $\alpha\in \mathbb A$ the cardinal $2^\alpha$ is a pre-fixpoint due to
$$\card H(2^\alpha)\leq1+\Sigma_{\beta\in \mathbb{B},\beta\leq 2^\alpha} (2^\alpha)^\beta\leq\Sigma_{\beta<\alpha} (2^\alpha)^\beta\leq\alpha\cdot(2^\alpha)^\omega=2^\alpha.$$
However, $H$ is not $\lambda$-accessible for any $\lambda$: choose
$\beta\in \mathbb B$ with $\beta\geq\lambda$, then since $\beta\in
H\beta$ we see that $H$ does not preserve the lambda-directed colimit
of $\beta$ as the union of all subsets of $\beta$ of cardinality less than $\lambda$.
\end{exa}

\section{Corecursive monads}
The iterative theories (or iterative monads) of C. Elgot \cite{elgot} were introduced as a formalization of iteration in an algebraic setting, and in \cite{ebt} completely iterative theories were studied. We first recall the concept of a completely iterative monad, and then introduce the weaker concept of a corecursive monad. The relationship between these two concepts is analogous to the relationship between cia's, see Remark 2.7, and corecursive algebras. The following definition is, for the base category {\bf Set}, equivalent to completely iterative theories, as shown in \cite{aamv}.

\begin{defi}(See \cite{aamv})\label{def ideal monad}
\begin{enumerate}
\item An {\it ideal monad} is a sixtuple $$\mathbb{S}=(S,\eta,\mu,S',\sigma,\mu')$$
consisting of a monad $(S,\eta,\mu)$, a subfunctor $\sigma:S'\rightarrow S$ (called the {\it ideal of $\mathbb S$} ) such that $S=S'+Id$ with injections $\sigma$ and $\eta$, and a natural transformation $\mu':S'S\rightarrow S'$ restricting $\mu$, i.e., with $\sigma\cdot\mu'=\mu\cdot\sigma S.$

\item An {\it equation morphism with parameters}  for $\mathbb S$ is a morphism $e:X\rightarrow S(X+Y)$, we call $X$ the variables and $Y$ the parameters of $e$.  It is called {\it ideal} if it factorizes through $\sigma_{X+Y}$. A solution of $e$ is a morphism $e^\dagger :X\rightarrow SY$ such that the following square commutes:
\begin{equation}\label{5.1}\vcenter{
\xymatrix@C+1pc{
X\ar[r]^-{e^\dagger}\ar[d]_e&SY\\
S(X+Y)\ar[r]_-{S[e^\dagger,\eta_Y]}&SSY\ar[u]_{\mu_Y}
}}
\end{equation}
\item An ideal monad is called {\it completely iterative} provided that every ideal equation morphism has a unique solution.
\end{enumerate}
\end{defi}

\begin{exa}(See \cite{aamv})
Let $H$  be an endofunctor of $\mathcal A$ such that for every object $Y$ a terminal coalgebra $TY$ of $H(-)+Y$ exists. Then the assignment $Y\mapsto TY$ yields a monad $(T,\eta,\mu)$, which is the monad of free cia's for $H$. This is an ideal monad w.r.t. $T'=HT$ and $\mu'=H\mu$. Moreover, this monad is completely iterative, indeed, the free completely iterative monad on $H$.

For example the set functor $HX=X\times X$ generates the free completely iterative monad $\mathbb T$ with
$$TY=\text{all binary trees with leaves labelled in $Y$.}$$
\end{exa}

\begin{defi}\label{def corecursive monad}
 Let $\mathbb S$ be an ideal monad. An {\it equation morphism (without parameters)} is a morphism $e:X\rightarrow SX$, and $e$ is called {\it ideal} if it factorizes through $\sigma_X$, i.\,e., there exist $e_0:X\rightarrow S'X$ such that the diagram below commutes:
$$
 \xymatrix{
 X\ar[r]^e\ar@{-->}[rd]_-{e_0}& SX\\
                       &S'X\ar[u]_{\sigma_X}
 }
 $$
The monad $\mathbb S$ is called {\it corecursive} if every ideal
equation morphism $e$ has a unique solution $e^\dagger: X \to SY$, i.\,e.,
the square below commutes:
\[
\xymatrix{
  X 
  \ar[r]^-{e^\dagger}
  \ar[d]_e
  &
  SY
  \\
  SX
  \ar[r]_-{Se^\dagger}
  &
  SSY
  \ar[u]_{\mu_Y}
}
\]
\end{defi}

\begin{exa}\label{5.4}
Examples of corecursive monads on {\bf Set}.
\begin{enumerate}
\item All the monads of Example \ref{some exams} are corecursive, as we will see in  Theorem \ref{M is free core monad} below.

\item All completely iterative monads are corecursive, e.g,
$$SY=\text{all finitely branching trees with leaves labelled in $Y$},$$
This is the free completely iterative monad on the functor $HX=\coprod_{n<\omega} X^n$.

\item Consider the monad $$RY=\text{all rational, finitely branching trees with leaves labelled in $Y$,}$$
where \emph{rational} means that the tree has up to isomorphism only finitely many subtrees. This is a corecursive monad that is neither free on any endofunctor, nor completely iterative.

\item More generally, every submonad of $S$ in item (2) containing the
  complete binary tree  is corecursive.
\end{enumerate}
\end{exa}

\begin{prop}\label{M is ideal monad}
The monad $\mathbb{M}=(M,\eta,\mu)$ of  free corecursive algebras (of Notation \ref{4.13}) is  ideal w.r.t. the ideal $M'=HM$ where $\sigma= \delta:HM\rightarrow M$ and $\mu'=H\mu:HMM\rightarrow HM$.
\end{prop}

\begin{proof}
(1) We prove $$MY=HMY+Y$$ with injections $\delta_Y$ and
$\eta_Y$. First apply Lemma \ref{generalized} to the corecursive
algebra $MY$ and $\eta_Y:Y\rightarrow MY$ to see that $HMY+Y$ is a
corecursive algebra, too. Using the freeness of $MY$ we see that for
the right-hand injection morphism $\mathsf{inr}:Y\rightarrow HMY+Y$
there exists a unique homomorphism
$\overline{\mathsf{inr}}:MY\rightarrow HMY+Y$ such that
$\overline{\mathsf{inr}} \o \eta_Y = \mathsf{inr}$. In order to prove
that $\overline{\mathsf{inr}}$ is an inverse for $[\delta_Y,\eta_Y]$
we consider the following diagram 
\begin{equation}\label{4}\vcenter{
\xymatrix{
Y\ar[r]^{\eta_Y}\ar[rd]^{\mathsf{inr}}\ar[rdd]_{\eta_Y}&MY\ar[d]^{\overline{\mathsf{inr}}}&&HMY\ar[ll]_-{\delta_Y}\ar[d]^{H\overline{\mathsf{inr}}}\\
&HMY+Y\ar[d]^{[\delta_Y,\eta_Y]}&&H(HMY+Y)\ar[d]^{H[\delta_Y,\eta_Y]}\ar[ll]^-{\mathsf{inl}\cdot H[\delta_Y,\eta_Y]}\\
&MY&&HMY\ar[ll]^-{\delta_Y}
}}
\end{equation}
The upper square and triangle commute by the definition of
$\overline{\mathsf{inr}}$, and the lower triangle and square obviously do. Thus $[\delta_Y,\eta_Y]$ is homomorphism, and the composition of the two squares clearly commutes. So $[\delta_Y,\eta_Y]\cdot\overline{\mathsf{inr}}=id_{MY}$ follows from the universal property of $\eta_Y$. Now the upper square of (\ref{4}) yields $\overline{\mathsf{inr}}\cdot\delta_Y=\mathsf{inl}$, consequently $(\overline{\mathsf{inr}}\cdot[\delta_Y,\eta_Y])\cdot\mathsf{inl}=\overline{\mathsf{inr}}\cdot\delta_Y=\mathsf{inl}$. Since also $(\overline{\mathsf{inr}}\cdot[\delta_Y,\eta_Y])\cdot\mathsf{inr}=\overline{\mathsf{inr}}\cdot\eta_Y=\mathsf{inr}$,
we conclude $\overline{\mathsf{inr}}\cdot[\delta_Y,\eta_Y]=id.$

(2) It remains to prove $\sigma\cdot\mu'=\mu\cdot \sigma S$, that is
$\delta\cdot H\mu=\mu\cdot \delta M$, which one obtains from Diagram~(\ref{mu is homo}).
\end{proof}

\begin{thm}
The monad $\mathbb M$ of free corecursive algebras is corecursive.
\end{thm}

\begin{proof}
We need to prove that every ideal equation morphism
$e=\delta_{X}\cdot e_0$ where $e_0:X\rightarrow HMX$  has a unique
solution in all algebras $(MY,\mu_Y)$. We prove a stronger statement:
if $a:HA\rightarrow A$ is a corecursive algebra, then $e$ has a unique
solution in $A$. That means that there exists a unique morphism $e^\dagger:X\rightarrow A$ such that the square
$$
\xymatrix{
X\ar[r]^{e^\dagger}\ar[d]_e&A\\
MX\ar[r]_{Me^\dagger}&MA\ar[u]_{\overline a}
}
$$
commutes, where $\overline a$ is  the corresponding Eilenberg-Moore algebra structure (the unique homomorphism of $H$-algebras from $MA$ to $A$ such that  $\overline a\cdot \eta_A=id_A$).
Form the equation morphism
$$ \xymatrix@C+2pc{\overline e\equiv MX\ar[r]^-{[\delta_X,\eta_X]^{-1}}&HMX+X\ar[r]^-{[HMX,e_0]}&HMX.}$$ 
There is a unique solution $s$ of $\overline e$:
\begin{equation}\label{5}\vcenter{
\xymatrix{
MX\ar[r]^{s}\ar@/_/[d]_{[\delta_{X},\eta_X]^{-1}}&A\\
HMX+X\ar@/_/[u]_{[\delta_{X},\eta_X]}\ar[d]_{[HMX,e_0]}\\
HMX\ar[r]_{Hs}&HA\ar[uu]_{a\ \ \ \ \ \ \ }
}}
\end{equation}
Inspecting the two coproduct components of $HMX+X$ separately, we see
that the commutativity of (\ref{5}) is equivalent to the following two
equations: 

\begin{align}
s\cdot\delta_{X}&=a\cdot Hs\label{6}\\
s\cdot\eta_X&=a\cdot Hs\cdot e_0\label{7}
\end{align}

Equation (\ref{6})  states that $s$ is a homomorphism. Also since $\eta:Id\rightarrow M$ is a natural transformation we have the commutative diagram

$$
\xymatrix{
X\ar[d]^{\eta_X}\ar[rr]^{\eta_X}&&MX\ar[d]^{M\eta_X}\\
MX\ar[d]_{s}\ar[rr]^{\eta_{MX}}&&MMX\ar[d]^{Ms}\\
A\ar[rr]_{\eta_A}&&MA
}
$$
which gives us $\overline a\cdot Ms\cdot M\eta_X\cdot \eta_X=\overline
a\cdot \eta_A\cdot s\cdot \eta_X=s\cdot\eta_X$. From this the
universal property of $\eta_X$ implies the equation
 \begin{align}
\overline a\cdot Ms\cdot M\eta_X=s. &\label{8}
\end{align}
We prove that $e^\dagger=s\cdot \eta_X$ is a solution of $e$ by inspecting the diagram

\begin{equation}\label{5.2}\vcenter{
\xymatrix@+1pc{
X\ar[r]^{\eta_X}\ar[ddd]_e\ar[rdd]_{e_0}\ar[rd]^{\mathsf{inr}}&MX\ar[d]^-{[\delta_X,\eta_X]^{-1}}\ar[rr]^s
&&A \ar@{<-} `u[l] `[lll]_{e^\dagger} [lll]\\
&HMX+X\ar[d]|(.4){[HMX,e_0]}&HA\ar[ur]^{a}\\
&HMX\ar[ld]^{\delta_{MX}}\ar[ru]_{Hs}\ar[r]_{HMe^\dagger}&HMA\ar[u]^{H\overline a}\ar[dr]_{\delta_{MA}}\\
MX\ar[rrr]_{Me^\dagger}&&&MA\ar[uuu]_{\overline a}
}}
\end{equation}

The  left-hand triangle commutes because $e$ is  an ideal equation. The  upper right-hand triangle  is Diagram (\ref{5}), the middle  triangle commutes because of Equation (\ref{8}) and $e^\dagger = s \cdot \eta_X$, and the other parts commute obviously.

To show uniqueness, suppose that $e^\dagger$ is an arbitrary solution of $e$. We take $s=\overline a\cdot Me^\dagger:MX\rightarrow A$,  and prove that $s$ is a solution of $\overline e$. First we note that $s$ is an algebra homomorphism for $H$, since both $\overline a$ and $Me^\dagger$ are, thus Equation (\ref{6}) holds for $s$. We also have $s \cdot \eta_X=e^\dagger$ using the naturality of $\eta$:
\[
s\cdot \eta_X = \overline a \cdot Me^\dagger \cdot \eta_X = \overline a \cdot \eta_A \cdot e^\dagger = e^\dagger.
\]
Now all the parts in Diagram (\ref{5.2}) commute except, possibly, the upper right-hand triangle. So this triangle commutes when precomposed with $\eta_X:X\rightarrow MX$ from which we conclude Equation (\ref{7}). Combining (\ref{6}) and (\ref{7}) we have, equivalently, that Diagram (\ref{5}) commutes, i.\,e., $s$ is a solution of $\overline e$.
\end{proof}

\section{Free Corecursive Monad}

In this section we prove that the corecursive monad $\mathbb M$ given by the free corecursive algebras for $H$ is a free corecursive monad on $H$. For that we need the appropriate concept of morphism:

\begin{defi}\hfill
\begin{enumerate}[label=\({\alph*}]
\item An {\it ideal monad morphism} from an ideal monad
  $(S,\eta^S,\mu^S,S',\sigma,\mu'^S)$ to an ideal monad
  $(U,\eta^U,\mu^U,U',\omega,\mu'^U)$ is a pair consisting of a monad
  morphism $\lambda:(S,\eta^S,\mu^S)\rightarrow (U,\eta^U,\mu^U)$ and
  a natural transformation $\lambda':S'\rightarrow U'$ with
  $\lambda\cdot\sigma =\omega\cdot \lambda'$.

\item Given a functor $H$, a natural transformation $\lambda:H\rightarrow S$ is called {\it ideal}  if it factors through $\sigma:S'\rightarrow S$.

\item By a {\it free corecursive monad} on an endofunctor $H$ is meant a corecursive monad ${\mathbb S}=(S, \mu, \eta, S',\sigma,\mu')$ together with an ideal natural transformation $\kappa:H\rightarrow S$ with the following universal property:
For every ideal natural transformation $\lambda:H\rightarrow
\overline{\mathbb S}$, where $ \overline{\mathbb S}$ is a corecursive
monad, there exists a unique ideal monad morphism $(\hat{\lambda}, \hat\lambda'):{\mathbb S}\rightarrow \overline{\mathbb S}$ such that the triangle below commutes:
$$\xymatrix{
H\ar[r]^\kappa\ar[rd]_\lambda&S\ar[d]^{\hat{\lambda}}\\
&\overline{S}
}
$$
\end{enumerate}
\end{defi}

\begin{rem}
Let $\mathsf{CMon}(\mathcal A)$ denote the category of corecursive monads and ideal monad morphisms. We have a forgetful functor  to $\mathsf{Fun}(\mathcal A,\mathcal A)$, the category of all endofunctors of $\mathcal A$, assigning to every corecursive monad $\mathbb S$ its ideal $ S'$. A free corecursive monad on $H\in \mathsf{Fun}(\mathcal A, \mathcal A)$ is precisely a universal arrow from $H$ to the above forgetful functor.
\end{rem}

\begin{exa}\label{exam ideal monad}
If $H$ has free corecursive algebras, then we have the corecursive monad $\mathbb M$ of Proposition \ref{M is ideal monad}. And the natural transformation $$\kappa\equiv H\stackrel{H\eta}\longrightarrow HM \stackrel{\delta}\longrightarrow M$$ (see Notation \ref{4.13}) is obviously ideal. We prove that $\kappa$ has the universal property:
\end{exa}

\begin{thm}\label{M is free core monad}
If an endofunctor $H$ has free corecursive algebras, then the corresponding monad $\mathbb M$ is the free corecursive monad on $H$.
\end{thm}

\begin{rem}
The proof is analogous to the corresponding theorem for free completely iterative monads, see \cite[Theorem~4.3]{m_cia}.
\end{rem}
\begin{proof}

For every corecursive monad $$\mathbb S=(S,\mu^S,\eta^S,S',\sigma,\mu'^S)$$ and every ideal natural transformation
$$
\xymatrix{
H\ar[r]^\lambda\ar[rd]_{\lambda'}&S\\
&S'\ar[u]_\sigma
}
$$
we are going to find an ideal monad morphism
$(\hat{\lambda},\hat\lambda'):\mathbb{M}\rightarrow \mathbb{S}$ with $\lambda=\hat{\lambda}\cdot\kappa$, and prove that it is unique.

\noindent\(a\ Every object $SA$, $A\in \mathcal A$, is considered as an algebra for $H$ via
\begin{align}\label{eq_5.4}
\rho_A\equiv HSA\stackrel{\lambda_{SA}}\longrightarrow SSA\stackrel{\mu_A^S}\longrightarrow SA.
\end{align}
We prove that $SA$ is  corecursive. Every equation morphism $e:X\rightarrow HX$ yields the following equation morphism $\overline e\equiv X\stackrel{e}\rightarrow HX\stackrel{\lambda_X}\longrightarrow SX$ w.r.t. the monad $\mathbb S$ and, and since $\lambda$ is an ideal natural transformation we see that $e$ an ideal equation morphism. We now prove that $e^\dagger$ is a  solution of $e$ in $SA$ if and only if it is the solution of $\overline e$ with respect to the corecursive monad $\mathbb{S}$.
Let $e^\dagger$ be a solution of $e$ in $SA$,  that is, the upper part of the diagram below commutes:
$$
\xymatrix{
X\ar[rr]^{e^\dagger}\ar[d]_{e}&&SA\\
HX\ar[r]^{He^\dagger}\ar[d]_{\lambda_X}&HSA\ar[ru]^{\rho_A}\ar[rd]_{\lambda_{SA}}&\\
SX\ar[rr]_{Se^{\dagger}}&&SSA\ar[uu]^{\mu_A^S}
}
$$
Since the  lower part commutes because of naturality of $\lambda$ and the right-hand triangle by definition of $\rho$, we see that  $e^\dagger$ is a solution of $\overline e$.

Conversely, if $e^\dagger$ is a solution of $\overline e$, then the outside of the above diagram commutes. Consequently, the upper part is commutative, showing $e^\dagger$ to be a solution of $e$ as desired. Thus, $SA$ is a corecursive algebra.\medskip

\noindent\(b\ There exists a unique homomorphism $\hat\lambda_A:MA\rightarrow SA$ of $H$-algebras such that $\hat\lambda_A\cdot \eta_A=\eta_A^S$. Now we show that $\hat\lambda$ is a natural transformation. Consider $f:A\rightarrow B$, then $Sf:SA\rightarrow SB$ is a homomorphism:
\begin{equation}\label{5.3}\vcenter{
\xymatrix{
HSA\ar[d]_{HSf}\ar[r]^{\lambda_{SA}}&SSA\ar[d]_{SSf}\ar[r]^{\mu_A}&SA\ar[d]^{Sf}\\
HSB\ar[r]_{\lambda_{SB}}&SSB\ar[r]_{\mu_B}&SB
}}
\end{equation}
The outside of the following diagram
$$
\xymatrix{
MA\ar[ddd]_{Mf}\ar[rr]^{\hat\lambda_A}&&SA\ar[ddd]^{Sf}\\
&A\ar[lu]^{\eta_A}\ar[d]_f\ar[ru]_{\eta_A^S}\\
&B\ar[ld]_{\eta_B}\ar[rd]^{\eta_B^S}\\
MB\ar[rr]_{\hat\lambda_B}&&SB
}
$$
commutes by the  universal property of $\eta_A$.\medskip

\noindent\(c\ We prove next that $\hat\lambda$ is a monad morphism. That is, the following diagrams are commutative:
$$
\xymatrix{
M\ar[rr]^{\hat\lambda}&&S&&MM\ar[r]^{\hat\lambda M}\ar[d]_{\mu}&SM\ar[r]^{S\hat\lambda}&SS\ar[d]^{\mu^S}\\
&Id\ar[ul]^{\eta}\ar[ur]_{\eta^S}& &&M\ar[rr]_{\hat\lambda}&&S
}
$$
The left-hand triangle commutes because of the definition of $\hat\lambda$. For the right hand square we note that
$S\hat\lambda_A$ is a homomorphism (cf. (\ref{5.3}) with $f=\hat{\lambda}_A$), all (components of) the other natural transformations are also clearly homomorphisms, and we have the following diagram
$$
\xymatrix{
MMA\ar[dd]_{\mu^M_A}\ar[rr]^{\hat\lambda_{MA}}&&SMA\ar[rr]^{S\hat\lambda_A}&&SSA\ar[dd]^{\mu_A^S}\\
&MA\ar[lu]^{\eta_{MA}}\ar[ru]_{\eta^S_{MA}}\ar[rr]_{\hat\lambda_A}&&SA\ar[ur]_{\eta^S_{SA}}\\
MA\ar@{=}[ur]\ar[rrrr]_{\hat\lambda_A}&&&&SA\ar@{=}[lu]
}
$$
The right-hand square commutes by naturality of $\eta^S$, hence the
outside square is commutative.\medskip

\noindent\(d\ Now we have to show that $\lambda=\hat\lambda\cdot\kappa=\hat\lambda\cdot \delta\cdot H\eta$ which follows from the commutativity of the diagram below:
$$
\xymatrix{
HA\ar[rr]^{H\eta_A}\ar[d]_{\lambda_A}&&HMA\ar[dl]^{\lambda_{MA}}\ar[dr]^{H\hat\lambda_A}\ar[rr]^{\delta_A}&&MA\ar[dd]^{\hat\lambda_A}\\
SA\ar[r]^{S\eta_A}\ar[drr]_{S\eta^S_A}&SMA\ar[dr]^{S\hat\lambda_A}&&HSA\ar[dl]_{\lambda_{SA}}\ar[rd]^{\rho_A}\\
&&SSA\ar[rr]_{\mu_A^S}&&SA\ar@{<-}`d[l]`[llll]^{id}[llllu]\\
}
$$
The left-hand upper part and the central one commute because $\lambda$ is an ideal natural transformation. The right-hand upper part commutes by the definition of $\hat \lambda$. The lower right-hand triangle commutes by the definition of $\rho$, see (\ref{eq_5.4}) and the lowest part commutes by the monad laws of $\mathbb S$.\medskip

\noindent\(e\ Next we show that $(\hat\lambda, \hat\lambda')$, where $\hat\lambda'
= {\mu'}^S \o \lambda'S \o H\hat\lambda$, is an ideal monad morphism. This follows from
$$
\xymatrix{
MA\ar[rrr]^{\hat\lambda}&&&SA\\
&&SSA\ar[ur]^{\mu^S_A}\\
HMA\ar[uu]_{\delta}\ar[r]_{H\hat\lambda_A}&HSA\ar[r]_{\lambda'_{SA}}\ar[ur]^{\lambda_{SA}}&S'SA\ar[u]^{\sigma_{SA}}\ar[r]_{\mu'^S}&S'A\ar[uu]_{\sigma}
}
$$

\noindent\(f\ It remains to prove that $(\hat \lambda,\hat\lambda')$ is
unique. Let $(\varphi, \varphi'): \mathbb M\rightarrow \mathbb S$ be
an ideal monad morphism with $\varphi \o \kappa=\lambda$. It is
sufficient to prove that $\varphi_A:MA\rightarrow SA$ is a
homomorphism of $H$-algebras w.r.t. the structure $\rho$ above and
$\varphi_A\o\eta_A=\eta_A^S$, then $\varphi_A=\hat \lambda_A$ of (b)
above. From that we derive $\varphi' = \hat\lambda'$ since $\sigma$ is
a monomorphism:
\[
\sigma \o \varphi' = \varphi \o \delta = \hat\lambda\o\delta = \sigma
\o \hat\lambda'. 
\]
 The equation $\varphi_A\o\eta_A=\eta_A^S$ follows from $\varphi$ preserving the units of the monad. And the fact that $\varphi_A$ is a homomorphism follows from the following diagram:
$$
\xymatrix{
HMA\ar[rr]^{\delta_A}\ar[rd]^{\kappa_{MA}}\ar[dd]_{H\varphi_A}&&MA\ar[dd]_{\varphi_A}\\
&MMA\ar[ru]^{\mu_A}\ar[d]_{(\varphi*\varphi)_A}\\
HSA\ar[r]_{\lambda_{SA}}&SSA\ar[r]_{\mu_A^S}&SA
}
$$
For the upper triangle see Remark~\ref{4.15} (and recall that $\kappa = \delta \cdot H\eta$), the right-hand square is the preservation of the monad multiplication, and for the left-hand one we use $(\varphi\o\kappa)M=\lambda M$ and the naturality of $\lambda$:
$$
\xymatrix{
HMA\ar[r]^{\kappa_{MA}}\ar[dd]_{H\varphi_A}\ar[rd]^{\lambda_{MA}}&MMA\ar[d]^{\varphi_{MA}}\\
&SMA\ar[d]^{S\varphi_A}\\
HSA\ar[r]_{\lambda_{SA}}&SSA
}
$$

\vspace*{-20pt}
\end{proof}

\begin{exa}\hfill
\begin{enumerate}
\item The functor $Id$ generates the free corecursive monad $$MY=\mathbb N\times Y+1,$$  see Example \ref{some exams}(1).
This is also the free completely iterative monad, since the functor $Id+Y$ has the terminal coalgebra $\mathbb N\times Y+1$.

\item The polynomial functor $H_\Sigma$ of a signature $\Sigma =(\Sigma_n)_{n<\omega}$ generates the free corecursive monad
\[\begin{array}{rcp{11.5cm}}
MY&= &all $(\Sigma+Y)$-trees in which only finitely many leaves
 are labelled in $Y$ (and other leaves labelled in $\Sigma_0$).
\end{array}
\]
 See Example \ref{some exams}(4).
\end{enumerate}
\end{exa}

\noindent Are there any other free corecursive monads than the monads $\mathbb M$ of free corecursive algebras?
Not for endofunctors of {\bf Set}:

\begin{prop}
If a set functor generates a free corecursive monad, then it has free corecursive algebras.
\end{prop}

\begin{proof}

\sloppypar Let $H:\mathbf{ Set}\rightarrow \mathbf{ Set}$ generate a free corecursive monad $\mathbb S=(S,\mu^S,\eta^S, S',\sigma,\mu')$, and let $\kappa:H\rightarrow S$ be the universal arrow. Following Theorem \ref{equivalence 1} we need to prove the existence of (a) arbitrary large pre-fixpoints and (b) a corecursive fixpoint.

The main technical statement is that the ideal $S'$ is naturally
isomorphic to $HS$. This proof is analogous to the same proof
concerning free completely iterative monads, see Sections 5 and 6 in
\cite{aamv}. We therefore omit it.

Ad (a). Since $SY=S'Y+Y=HSY+Y$ for every set $Y$, we see that $SY$ is
a pre-fixpoint of cardinality at least $\card Y$.

Ad (b). The isomorphism $\sigma_{\emp}:HS\emp\rightarrow S\emp$
defines a corecursive algebra for $H$. To prove this, consider an
arbitrary equation morphism $e:X\rightarrow HX$ and form the equation
morphism $\overline e=\kappa_X\cdot e:X\rightarrow SX$. Then solutions of
$\overline e$ w.r.t $\mathbb S$ (in $S\emp$) are in bijective
correspondence with solutions $e$ in the algebra $S\emp$. This is easy
to prove, the details are as in the of proof of Theorem 6.1 of
\cite{m_cia}.
\end{proof}

\section{Hyper-Extensive Categories}
Recall that for the more general case of recursion with parameters the equational properties of $\dagger$ are captured by iteration theories of S.~Bloom and Z.~\'{E}sik \cite{be}. Recently \emph{functoriality} was ``added" to these properties; functoriality states that for two equation morphisms with parameters $e: X \to S(X+Y)$ and $f: Z \to S(Z+Y)$ (cf.~Definition~\ref{def ideal monad}) we have
\[
\vcenter{
  \xymatrix@R-1pc{
    X \ar[r]^-e \ar[dd]_h
    &
    S(X+Y)
    \ar[dd]^{S(h+Y)}
    \\ \\
    Z \ar[r]_-f
    &
    S(Z+Y)
    }
  }
  \qquad\Longrightarrow\qquad
  \vcenter{
    \xymatrix@R-1pc{
      X \ar[rd]^{e^\dagger} \ar[dd]_h
      \\
      & SY
      \\
      Z\ar[ru]_{f^\dagger}
      }
    }
\]
Being an implication, this is not equational if one takes, as in \cite{be}, the category of signatures as the base category. Instead, in \cite{amv_em2} the presheaf category
$$\mathbf{Set}^\mathbb{F}\ \text{ (sets in context)}$$
where $\mathbb F$ is the category of finite sets and functions, was
suggested as a base category. Equivalently, this is the category of
all finitary endofunctors on $\mathbf{Set}$.  Then functoriality is an
equational property in the sense of Kelly and Power~\cite{kp93}, and the functorial iteration theories are called {\it Elgot Theories} in \cite{amv_em2}. It follows from the results
in~\cite{be} that all equational properties of $\dagger$ in Domain
Theory are precisely captured by the concept of iteration theory. \smnote{I tried to clarify with the following sentence (also below Theorem).} More precisely, every equation that holds for a parametrized fixpoint operator $\dagger$ given by least fixpoints in a category of domains follows from the axioms of iteration theories (see e.g.~Simpson and Plotkin~\cite{sp00}). 

We have proved in~\cite{amv_em2} that Elgot theories are monadic over sets in context:
%
%
\begin{thm}[\cite{amv_em2}]\label{eilenberg-moore alg for M=elgot theo}
 Form the monad $\mathcal M$ on $\mathbf{Set}^\mathbb F$ by assigning to every set in context $H$ the free iterative theory on $H_\bot=H(-)+1$ of C.~Elgot \cite{elgot}. Then the Eilenberg-Moore algebras for  $\mathcal M$ are precisely the  Elgot theories.
\end{thm}
This result implies, using the results of Kelly and Power~\cite{kp93}, that Elgot theories are equational over sets in context, and we gave one axiomatization (that includes functoriality) in~\cite{amv_em2}.

\begin{exa}
  The polynomial set functor $H_\Sigma$ of Example~\ref{some exams}(4) defines a set in context (that we also denote by $H_\Sigma$) by a domain restriction to $\mathbb F$. The corresponding iterative theory $\mathcal M(H_\Sigma)$ is given by
  \[
  X \mapsto \text{all rational trees labelled in $\Sigma + X +\{\bot\}$}.
  \]
  This is a subtheory of the theory $\mathcal T_\Sigma$ of all trees labelled in $\Sigma + X +\{\bot\}$, which is the free continuous theory (see Example~\ref{ex 3}(c)).
\end{exa}

It was proved by Bloom and \'Esik in~\cite{be} that the equational properties of the operation $\dagger$ (of solving recursive equations) of the above theory $\mathcal T_\Sigma$ are precisely the equational properties that $\dagger$ has in an impressive number of applications of iteration. Thus, the axiomatization of these properties in~\cite{be} can be understood as the summary of equational properties that $\dagger$ is expected to have in applications. 

For every finitary set functor $H$ there exists a signature $\Sigma$ such that $H$ is a quotient of $H_\Sigma$ (see~\cite{at}). Therefore, $\mathcal M(H)$ is a quotient theory of the theory $\mathcal M(H_\Sigma)$ of rational trees. Thus, the equational properties of $\dagger$ in all free Elgot theories $\mathcal M(H)$ for finitary set functors $H$ are determined by those of $\dagger$ in rational trees.

In the present section we provide the first steps to an analogous result for iteration without parameters. We introduce  \emph{finitary corecursive monads} as the analogy of Elgot's iterative theories, and we prove that every finitary endofunctor on {\bf Set} generates a free finitary corecursive monad. Let $\mathcal M^*$ be the monad on $\mathbf{Set}^\mathbb F$ given by
$$\mathcal M^*(H)= \text{free finitary corecursive monad on }H_\bot.$$
Then we prove that the Eilenberg-Moore algebras for $\mathcal M^*$ are precisely the {\it Bloom theories}, i.\,e., theories with an operation $\dagger$ satisfying the equational properties that hold in non-parametric iteration. We list some of these properties. It is an open problem whether our list is complete.

In lieu of {\bf Set} we work, more generally, in a locally finitely presentable category. Thus in lieu of theories we work with finitary monads (in analogy to iterative monads of \cite{elgot}). For most of the results we need to assume the category we work with is hyper-extensive. We now start by recalling this concept from \cite{abmv_how}.

\begin{defi}(See \cite{abmv_how})
A locally finitely presentable category $\mathcal A$ is called {\it hyper-extensive} if every object is a coproduct of connected objects, i.\,e., objects $A$ such that $\mathcal A(A,-)$ preserves coproducts.
\end{defi}

\begin{rem}
In \cite{abmv_how} the definition is different, but Theorem 2.7 of \cite{abmv_how} states that the present formulation is equivalent. Every hyper-extensive category is extensive, i.~e., coproducts are
\begin{enumerate}[label=\({\alph*}]
\item disjoint (coproduct injections are monic and pairwise intersections always yield $0$)
and

\item universal (preserved by pullback along any morphism).
\end{enumerate}
Moreover, in hyper-extensive categories we have
\begin{enumerate}[label=\({\alph*},resume]
\item given pairwise disjoint monics $a_i:A_i\rightarrow B$, $i\in  \mathbb N$, if each $a_i$ is coproduct injection then so is $[a_i]:\coprod_{i\in \mathbb N}A_i\rightarrow B$.
\end{enumerate}
For locally finitely presentable categories (a)--(c) are equivalent to hyper-extensivity.
\end{rem}

\begin{exa}
Sets, posets, graphs, and every presheaf category are hyper-extensive. Given a signature $\Sigma$ the category of $\Sigma$-algebras is hyper-extensive iff all arities are 1.
\end{exa}

\begin{rem}
In a hyper-extensive category a monad $\mathbb S=(S,\mu, \eta)$ is
{\it ideal} (see Definition \ref{def ideal monad}) iff that $S$
is a coproduct $S=S'+Id$ with injections $\sigma:S'\rightarrow S$ and
$\eta:Id\rightarrow S$ and the multiplication has a
restriction $$\mu':S'S\rightarrow S'.$$
Thus, in this setting ``ideal'' is a property not an additional structure of a monad. 
\end{rem}

\begin{nota}
$\mathsf{Mon}_i(\mathcal A)$ denotes the category of ideal monads and ideal monad morphisms, i.e., morphisms $\alpha:\mathbb S\rightarrow \mathbb T$ for which a restriction to the ideals exist: we have 
\[
\alpha=(\xymatrix@1@C+2pc{
  S = S' +Id 
  \ar[r]^-{\alpha' + Id}
  & T' + Id = T
})
\]
for a natural transformation $\alpha':S'\rightarrow T'$.
\end{nota}

\begin{rem}
We shall prove below that every corecursive monad $\mathbb S$ on a hyper-extensive category has solutions for all, not only ideal, equation morphisms. For that we need to specify an element of $S0$ which then serves for defining solutions of non-ideal equations such as $x=x$. In the following definition $1$ denotes the terminal object of $\mathcal A$ and $0$ the initial one. The unique morphism from $0$ to $X$ is denoted by $!:0\rightarrow X$. Analogously $!:X\rightarrow 1$.
\end{rem}

\begin{defi}
A {\it strict endofunctor} is an endofunctor $H$ together with a morphism $\bot: 1\rightarrow H0$. A monad is {\it strict} if its underlying endofunctor is. A natural transformation $\alpha: H\rightarrow K$ between strict functors is called strict if $\alpha_0$ preserves $\bot$.

Every strict endofunctor has a special global element in every $HX$: it is the composite of $\bot:1\rightarrow H0$ and $H!:H0\rightarrow HX$. We denote it again by $\bot:1\rightarrow HX$.
\end{defi}

\begin{rem}
We now recall from \cite{abmv_how} the concept of a strict solution and the fact that every equation morphism has a unique strict solution. In \cite{abmv_how} equation morphisms with parameters $Y$ (see Definition \ref{def ideal monad}) were considered, here we restrict ourselves to $Y=0$.
\end{rem}

\begin{defi}{\rm (See \cite{abmv_how})}\label{def_derived subobjects}
(a) For every equation morphism $e:X\rightarrow SX$ we denote by $$i_\infty:X_\infty\rightarrow X$$
the intersection of the \emph{derived subobjects} $i_1,\, i_1\cdot i_2,\,i_1\cdot i_2\cdot i_3,\,\cdots$ obtained by forming recursively pullbacks as follows:
\begin{equation}\label{dia_derived subobjects}\vcenter{
\xymatrix{
\ar@{}[r]|(.65){\objectstyle \cdots}& X_3\ar[d]_{e_3}\ar@{>->}[r]^-{i_3}&{X_2}\ar[d]_{e_2}\ar@{>->}[r]^-{i_2}&{X_1}\ar[d]_{e_1}\ar@{>->}[r]^-{i_1}&X=X_0\ar[d]_{e}\\
{\textstyle\cdots}\ \ar@{>->}[r]_-{i_3}&X_{2}\ar@{>->}[r]_-{i_2}&X_{1}\ar@{>->}[r]_-{i_1}&X\ar[r]_{\eta_X}&SX
}
}
\end{equation}
\end{defi}

\begin{rem}
$X_\infty$ represents those variables for which solutions of $e$ have ``difficulties" assigning a value. For example, if $e$ represents the iterative equation $x=x$ or the system $$x=y,\quad y=x$$
then $X=X_\infty$. We resolve the difficulties by assigning the value $\bot$ to such variables:
\end{rem}

\begin{defi}(See \cite{abmv_how})
Let $e:X\rightarrow SX$ be an equation morphism. A solution $e^\dagger:X\rightarrow SY$ is called {\it strict} if its restriction to $X_\infty$ factorizes through $\bot$:

$$ \xymatrix{
 X_\infty\ar[r]^{!}\ar[d]_{i_\infty}&1\ar[d]^{\bot}\\
 X\ar[r]_{e^\dagger}&SY
 } $$
\end{defi}

\begin{thm}{\rm(}See \cite{abmv_how}{\rm)}\label{the_ equations have strict solutions}
Let $\mathbb S$ be a strict, corecursive monad on a hyper-extensive category. Then every equation morphism $e:X\rightarrow SX$ has a unique strict solution $e^\dagger :X\rightarrow SY$, for every object $Y$.

\end{thm}

In \cite{abmv_how} we proved this for completely iterative monads, the proof for the corecursive monads is the same.

\section{Finitary Bloom algebras and monads}

In this section we investigate the variant of iteration in which only equation morphisms $e:X\rightarrow SX$ with finitely presentable objects $X$ (of variables) are considered.

Throughout this section we assume that the base category is locally finitely presentable and hyper-extensive. And a finitary endofunctor $H$ is given.

\begin{defi}\hfill
\begin{enumerate}[label=\({\alph*}]
\item  An algebra $a:HA\rightarrow A$ is said to be {\it finitary corecursive} if for every coalgebra $e:X\rightarrow HX$ with $X$ finitely presentable there exists a unique solution, i.\,e., a unique coalgebra-to-algebra morphism $e^\dagger:X\rightarrow A$.

\item A {\it finitary Bloom algebra} is a triple $(A,a,\dagger)$ where $a:HA\rightarrow A$ is an algebra and $\dagger$ an operation which to every $e:X\rightarrow HX$, $X$ finitely presentable, assigns a solution $e^\dagger: X\rightarrow A$ subject to functoriality: for every coalgebra homomorphism $h:(X,e)\rightarrow (X',e')$ with $X$ and $X'$ finitely presentable the following triangle commutes:
$$\xymatrix{
X\ar[rd]_{e^\dagger}\ar[rr]^h&&X'\ar[ld]^{(e')^\dagger}\\
&A
}$$

\item Homomorphisms are defined analogously to Definition \ref{preserv
    solu}.
\end{enumerate}
\end{defi}

\begin{rem}\label{rem_finitely cores are finitely Bloom}\hfill
\begin{enumerate}[label=\({\alph*}]
\item  Every finitary corecursive algebra is a finitary Bloom algebra: the functoriality follows from the uniqueness of solutions.

\item Lemmas~\ref{homos are solution prese} and~\ref{A is B then B is B} hold also for finitary
  Bloom algebras. 
\end{enumerate}
\end{rem}

\begin{exa}
Consider unary algebras in $\mathbf{Set}$, that is, $H=Id$. 
\begin{enumerate}[label=\({\alph*}]
\item 
An algebra $a: A\to A$ is finitary corecursive iff $a$ has a unique fixpoint $t = a(t)$. Indeed, for every equation morphism $e:X\to X$ with $X$ finite the unique solution is $e^\dag=\mathsf{const}_t$.

Thus the algebra $\mathbb Z^+$ of integers with $\infty$ where the unary operation is successor (and $\infty$ is its fixpoint) is finitary corecursive. But not corecursive: consider the system of equations given by $x_i=a(x_{i+1})$ for $i\in \mathbb N$. It has more that one solution in $\mathbb Z^+$, e.g., $x_i\mapsto -i$ and $x_i \mapsto \infty$ are solutions.
\item An algebra $a: A \to A$ is a finitary Bloom algebra iff $a$ has a fixpoint -- this is the same as in Example~\ref{ex 3}(b).
\end{enumerate}
\end{exa}

\begin{rem}\hfill
\begin{enumerate}[label=\({\alph*}]
\item Recall from \cite{amv_atwork} the concept of an iterative algebra: it is an algebra $a:HA\rightarrow A$ such that every equation morphism $e:X\rightarrow HX+A$ with $X$ finitely presentable has a unique solution. That is, the algebra $[a,A]:HA+A \rightarrow A$ for $H(-)+A$ is finitary corecursive. Every iterative algebra is obviously finitary corecursive (for $H$). An example of a finitary corecursive algebra that is not iterative is the algebra of all binary trees with finitely many leaves, all of which are labelled in $Y$, see Example \ref{some exams}(3).

\item Recall further from \cite{amv_atwork} that the
  category $$\mathsf{Coalg}_f\,H$$ of all coalgebras on finitely
  presentable objects of $\mathcal A$ is filtered, and the filtered
  colimit of the forgetful functor to $\mathcal A$, $$R={\rm
    colim}\{X;\ (X,e)\in \mathsf{Coalg}_f\,H\}$$ carries the structure
  of a coalgebra $i:R\to HR$. This structure is an isomorphism, and
  its inverse $\rho:HR\rightarrow R$ is the initial iterative
  algebra. Consequently, $R$ is a finitary Bloom algebra as
  well. Indeed:
\end{enumerate}
\end{rem}

\begin{prop}\label{R is i.f.c.a}
$R$ is the initial finitary corecursive algebra.
\end{prop}

\begin{proof}
We know that $R$ is finitary corecursive because it is even iterative. Let $(A,a)$ be a finitary Bloom algebra. Given a (solution-preserving) homomorphism $h:R\rightarrow A$, for every $e:X\rightarrow HX$ in $\mathsf{Coalg}_f\,H$ the triangle
$$\xymatrix{
&X\ar[ld]_{e^\sharp}\ar[rd]^{e^\dag}\\
R\ar[rr]_h&&A\\
}$$
commutes, where $e^\sharp$ denotes the solution in $R$. As proved in \cite{amv_atwork}, these morphisms $e^\sharp$ form the colimit cocone of $R$ (as a colimit of the forgetful functor of $\mathsf{Coalg}_f\,H$). Thus, the above triangles determine, since $\dag$ is functorial, a unique morphism $h$ which is solution-preserving. It remains to prove that $h$ is a homomorphism. For that recall from \cite{amv_atwork} that the algebra structure $\rho :HR\rightarrow R$ is defined as the inverse of the unique isomorphism $i:R\rightarrow HR$ with $$i\cdot e^\sharp =He^\sharp\cdot e\ \text{for all }\ e:X\rightarrow HX \in \mathsf{Coalg}_fH.$$
Thus in order to prove $h\cdot \rho=a\cdot Hh$ we use that $h=a\cdot Hh\cdot i:R\rightarrow A$ which follows from the fact that $e^\sharp$ are collectively epic: $h\cdot e^\sharp=e^\dag=a\cdot He^\dag\cdot e=a\cdot Hh\cdot He^\sharp\cdot e=a\cdot Hh\cdot i\cdot e^\sharp$.
\end{proof}

\begin{exa}\hfill
\begin{enumerate}[label=\({\alph*}]
\item For $HX=X\times X+1$ the algebra $R$ consists of all binary trees that are rational, i.\,e., have finitely many subtrees up to isomorphism, see \cite{ginali}.

\item More generally, given a finitary signature $\Sigma$ the polynomial functor $H_\Sigma:\mathbf{Set}\rightarrow \mathbf{Set}$ with
$$H_\Sigma X=\Sigma_0+\Sigma_1\times X+\Sigma_2\times X^2+\cdots$$
has the initial iterative algebra
$$R_\Sigma=\text{all rational $\Sigma$-trees}.$$

\item For the finite power-set functor $R$ is the algebra of all rational, finitely branching, strongly extensional trees in the sense of J. Worrell \cite{w}.
\end{enumerate}
\end{exa}

\begin{thm}
Every object $Y$ generates a free finitary Bloom algebra. This is the
coproduct $R\oplus FY$ of $R$ and the free algebra $FY$ in $\mathsf{Alg}\,H$.
\end{thm}
\begin{proof}
The category $\mathsf{Alg}\,H$ is locally finitely presentable, see \cite{ar}, thus $\oplus$ exists in $\mathsf{Alg}\,H$. The coproduct injection $\mathsf{inl}:R\rightarrow R\oplus FY$ induces a finitary Bloom algebra structure on $R\oplus FY$, see Lemma  \ref{A is B then B is B}. And we have a canonical morphism $\mathsf{inr}\cdot \eta_Y:Y\rightarrow R\oplus FY$ in $\mathcal A$.

The universal property of $\mathsf{inr}\cdot \eta_Y$ is clear:
let $(B,b,\ddag)$ be a finitary Bloom algebra and $g:Y\rightarrow B$ a morphism in $\mathcal A$.
$$\xymatrix{
R\ar[r]^-{\mathsf{inl}}\ar[rd]_-f&R\oplus FY\ar[d]_-{[f,\overline{g}]}&FY\ar[l]_-{\mathsf{inr}}\ar[ld]_-{\overline{g}}&Y\ar[l]_-{\eta_Y}\ar[lld]^-g\\
&B
}$$
We have a unique homomorphism $\overline g:FY\rightarrow B$ and, due to Proposition \ref{R is i.f.c.a}, a unique solution-preserving homomorphism $f:R\rightarrow B$. The unique homomorphism $[f,\overline g]:R\oplus FY\rightarrow B$ is solution-preserving since $f$ is and the Bloom algebra structure of $R\oplus FY$ is induced by $\mathsf{inl}$.
\end{proof}

\begin{cor}
$R\oplus FY$ is a free finitary corecursive algebra on $Y$.
\end{cor}
Indeed, the proof that solutions in $R\oplus FY$ are unique is completely analogous to that of Theorem \ref{free cor exists}.

\begin{exa}
Let $\Sigma$ be a finitary signature. The free finitary corecursive algebra $M^*Y=R_\Sigma\oplus F_\Sigma Y$ is the algebra of all rational trees with finitely many leaves labelled in $Y$ and all other nodes (with $n$ successors) labelled by an $n$-ary operation in $\Sigma$ for $n=0, 1, 2, \ldots$.
\end{exa}

\begin{nota}
The monad of free finitary Bloom algebras is denoted by $\mathbb M_H^*$. It is defined on objects by assigning to $Y$ the underling object of $R\oplus FY$. In the case of $H=H_\Sigma$ we write $\mathbb M^*_\Sigma$.
\end{nota}

\begin{thm}
The category of finitary Bloom algebras is isomorphic to the Eilenberg-Moore category of ${\mathbb M_H^*}$.
\end{thm}

The proof is analogous to that of Theorem \ref{the bloom alg=eilenberg moor off M_H}.

\begin{rem}
Analogous to the concept of a corecursive monad, see Definition \ref{def corecursive monad}, we call an ideal monad $\mathbb S$ \emph{finitary corecursive} if every ideal equation morphism $e:X\to SX$ with $X$ finitely presentable has a unique solution $e^\dag:X\to SY$. The following result is completely analogous to Proposition \ref{M is ideal monad} and Theorem \ref{M is free core monad}.
\end{rem}

\begin{thm}
For every finitary endofunctor $H$ the monad $\mathbb M^*_H$ of free finitary corecursive algebras is a free finitary corecursive monad on $H$. This monad is ideal with the ideal $HM$ and $\mu'=H\mu$.
\end{thm}

\begin{nota}\hfill
\begin{enumerate}[label=\({\alph*}]
\item Let $\mathbb F$ be a full subcategory of $\mathcal A$ representing all finitely presentable objects. The functor category $\mathcal A^{\mathbb F}$ is equivalent to the category of all finitary endofunctors of $\mathcal A$.

\item $\mathcal A_\bot^\mathbb F$ denotes the non-full subcategory of all strict finitary endofunctors (and strict natural transformations). The embedding $\mathcal A_\bot^\mathbb F\hookrightarrow\mathcal A^\mathbb F$ has a left adjoint $H\mapsto H_\bot$ where $H_\bot X=HX+1$.

\item We denote by $$\mathcal M^*$$ the monad on $\mathcal A^\mathbb F$ given by free finitary corecursive monads: $$\mathcal M^*:H\mapsto M^*_{H_\bot}.$$
More precisely $\mathcal M^*$ is the monad obtained by the composite adjoint situation
$$\xymatrix{
\mathcal{A}^\mathbb {F}
\ar@<5pt>[r]^{(-)_\bot}
\ar@{}[r]|-{\perp}
&
\mathcal{ A}^\mathbb{F}_\bot
\ar@<5pt>[l]
\ar@<5pt>[r]^-{M^*_{(-)}}
\ar@{}[r]|-\perp
&
\mathsf{FC}_{\bot}(\mathcal{ A})
\ar@<5pt>[l]^-U
}$$
where $U$ is the forgetful functor of the category
$\mathsf{FC}_{\bot}(\mathcal{ A})$ of all strict finitary corecursive
monads.
\end{enumerate}
\end{nota}

\begin{defi}
A \emph{Bloom monad} on $\mathcal A$ is an Eilenberg-Moore algebra for the monad $\mathcal M^*$.
\end{defi}

\begin{rem}
Following Theorem \ref{eilenberg-moore alg for M=elgot theo}, this concept is, for $\mathcal A=\mathbf{Set}$, completely analogous to iteration theories of Bloom and \'Esik: whereas iteration theories formalizes equational properties of parametrized iteration, Bloom monads on $\mathbf{Set}$ formalize equational properties of non-parametrized iteration. But what are Bloom monads?

\begin{enumerate}
\item Every Bloom monad is a finitary monad on $\mathcal A$. Indeed, let
$\mathcal F$ be the free-monad on $\mathcal A^\mathbb{F}$: to every
finitary endofunctor $H$ it assigns the free monad $F_H$ on $H$. It is
well-known that the Eilenberg-Moore algebras for $\mathcal F$ are
precisely the finitary monads on $\mathcal A$,
see~\cite{lack}.

For every $H$ in $\mathcal A^\mathbb F$ we have the unique monad morphism $$\phi_H:\mathcal F(H)\to \mathcal M^*(H)=M^*_{H_\bot}$$ given by the universal property of $\mathcal F(H)$. These morphisms form components of a monad morphism $\phi:\mathcal F\to \mathcal M^*$ (over $\mathcal A^\mathbb F$). Thus, every Eilenberg-Moore algebra for $\mathcal M^*$ is automatically one for $\mathcal F$, too.

\item Every Bloom monad $\mathbb S=(S,\eta,\mu)$ comes equipped with an operation $\dagger$ assigning to every equation morphism $e:X\to SX$ with $X$ finitely presentable and every object $Y$ a morphism $e^\dag:X\to SY$ which is a solution:
$$e^\dag=\mu_Y\cdot Se^\dag\cdot e.$$
Indeed, the Eilenberg-Moore structure
$$\sigma:M^*_{S_\bot}\to S$$
is a monad morphism, and we have also the universal arrow (see Example \ref{exam ideal monad})
$$\xymatrix{\overline \kappa\equiv S\ar[r]^-{\mathsf{inl}}&S_\bot\ar[r]^-\kappa&M^*_{S_\bot}}$$
which, due to the unit law of $\sigma $, fulfils
\begin{equation}\label{eq g.kappa=id}
\sigma\cdot \overline \kappa=id_S.
\end{equation}
For every equation morphism $e:X\to SX$ the unique strict solution (see Theorem \ref{the_ equations have strict solutions}) of $\overline{\kappa}_X\cdot e:X\to M^*_{S_\bot}X$ w.r.t. $\mathbb M^*_{S_\bot}$ is denoted by $e^\ddag:X\to M^*_{S_\bot}Y$. Then
$$e^\dag=\sigma_Y\cdot e^\ddag:X\to SY$$
is a (canonical) solution of $e$ w.r.t. $\mathbb S$. Indeed, the following diagram, where $\overline \mu$ denotes the multiplication of $\mathbb M^*_{S_\bot}$, commutes: 
$$\xymatrix{
X\ar[rr]^-{e^\ddag}\ar[dd]_-e&&M^*_{S_\bot}Y\ar[r]^{\sigma_Y}&SY\\
&M^*_{S_\bot}X\ar[r]^-{M^*_{S_\bot}e^\ddag}&M^*_{S_\bot}M^*_{S_\bot}Y\ar[u]_-{\overline \mu_Y}\ar[rd]_-{(\sigma *\sigma)_Y}\\
SX\ar[ur]_-{\overline{\kappa}_X}\ar[rrr]_-{S(\sigma_Y\cdot e^\ddag)}&&&SSY\ar[uu]_-{\mu_Y}
}$$
The left-hand part commutes since $e^\ddag$ is a solution of $\overline{\kappa}_X\cdot e$. The right-hand part (with $\sigma*\sigma=S\sigma\cdot\sigma M^*_{S_\bot}$) commutes because $\sigma:\mathbb M^*_{S_\bot}\to \mathbb S$ is a monad morphism. And to prove the lower part, we just need to verify
$$Se^\ddag=\sigma_{M^*_{S_\bot}Y}\cdot M^*_{S_\bot}e^\ddag\cdot \overline \kappa_X$$
which follows easily from $\sigma_X\cdot \overline \kappa_X=id$ and
the naturality of $\sigma$.
\end{enumerate}
\end{rem}

\begin{rem}\label{re equational prop}
(a) The operation $\dag$ above satisfies all the equational laws that the formation of strict solutions in all finitary corecursive monads satisfies. This follows from the fact that the Bloom monad $(\mathbb S,\dag)$ is by definition a quotient algebra (for $\mathcal M^*$) of the finitary corecursive monad $(\mathbb M^*_{S_\bot},\ddag)$.

(b) For the base category $\mathbf{Set}$ we can say more. Since $S$ is
finitary, there exists a finitary signature $\Sigma$ such that $S$ is
a quotient of $H_\Sigma$ (see~\cite{at}). Let $\Sigma_\bot=\Sigma\cup \{\bot\}$ be the extension by a nullary operation $\bot$. Then $S_\bot$ is a quotient of $H_{\Sigma_\bot}$. Indeed, there exists a signature $\Gamma$ such that for suitable natural transformations $\alpha_1,\alpha_2:H_\Gamma\to H_\Sigma$ we have a coequalizer
$$\xymatrix{H_\Gamma\ar@<2pt>[r]^-{\alpha_1}\ar@<-2pt>[r]_-{\alpha_2}&H_{\Sigma_\bot}\ar[r]^e&S_\bot},$$
see \cite{amm12}. The functor $\mathbb M^*_{(-)}$ of free  finitary corecursive monads is a left adjoint, thus, it preserves coequalizers. Therefore $\mathcal{M}S_\bot$ is a quotient of $\mathcal{M}H_{\Sigma_\bot}$.

Consequently, a Bloom monad in $\mathbf{Set}$ is a finitary monad with a solution operation $\dag$ satisfying all the equational laws that the operation of unique strict solutions for the rational-tree monads $\mathbb M^*_{\Sigma_\bot}$ satisfies.
\end{rem}

In the following example we list some equational properties of $\dag$
in Bloom monads. It is an open problem whether this list is complete
in the sense that every equational property of $\dag$ holding in all
Bloom monads can be derived from the properties stated.

\begin{exa}
Equational properties of $\dag$ in Bloom monads. We use the terminology of the monograph \cite{be}.
\begin{enumerate}[label=\({\alph*}]
\item Fixpoint identity. This is the equation
$$e^\dag=\mu_Y\cdot Se^\dag\cdot e$$
of Definition \ref{def corecursive monad}.

\item Functoriality (called functorial dagger implication in~\cite{be}). For every coalgebra homomorphism
$$\xymatrix{
X\ar[r]^e\ar[d]_-h&SX\ar[d]^-{Sh}\\
\overline{X}\ar[r]_-{\overline{e}}&S\overline X
}$$
we have
$$e^\dag={\overline e}^\dagger\cdot h.$$
To prove this recall that we can restrict ourselves to finitary
corecursive monads where strict solutions are unique (Remark
\ref{re equational prop}(a)).
Thus, it is sufficient to observe that $\overline{e}^\dag\cdot h$ solves $e$:
$$\mu_Y\cdot S(\overline{e}^\dag\cdot h)\cdot e=\mu_Y\cdot S\overline e^\dag\cdot \overline e\cdot h=\overline e^\dag\cdot h$$
and that it is strict. The latter follows from the fact that $\overline e^\dag$ is strict for $\overline e$ and the subobjects $i_n:X_n\to X_{n-1}$ of Definition \ref{def_derived subobjects} for $e$ are related to the corresponding subobject $\overline i_n:\overline X_n\to \overline X_{n-1}$ for $\overline e$ by morphisms $h_n:X\to \overline X_n$ such that the squares below commute:
\[
\xymatrix{
  \cdots\  
  \ar@{>->}[r]^{i_3} 
  &
  X_2
  \ar@{>->}[r]^{i_2} 
  \ar[d]^{h_2}
  &
  X_1
  \ar@{>->}[r]^{i_1} 
  \ar[d]^{h_1}
  &
  X
  \ar[d]^{h = h_0}
  \\
  \cdots\ 
  \ar@{>->}[r]_{\overline i_3} 
  &
  \overline X_2
  \ar@{>->}[r]_{\overline i_2}
  &
  \overline X_1
  \ar@{>->}[r]_{\overline i_1}
  &
  \overline X 
}
\]
For example, $h_1$ is the unique morphism for which the following diagram commutes:
\[
\xymatrix{
  X_1
  \ar[ddd]_{e_1}
  \ar@{>->}[rrr]^-{i_1}
  \ar[rd]^{h_1}
  &&&
  X
  \ar[ddd]^{e}
  \ar[ld]_h
  \\
  &
  \overline X_1
  \ar@{>->}[r]^-{\overline i_1}
  \ar[d]_{\overline e_1}
  &
  \overline X
  \ar[d]^{\overline e}
  \\
  &
  \overline X
  \ar@{>->}[r]_-{\eta_X}
  &
  S\overline X
  \\
  X
  \ar@{>->}[rrr]_-{\eta_X}
  \ar[ru]^h
  &&&
  SX
  \ar[lu]_{Sh}
}
\]

satisfying $h_{n-1}\cdot i_n=\overline i_n\cdot h_n$ (where $h_0=h$). Therefore, we obtain $h_\infty :X_\infty\to\overline X_\infty$ such that $h\cdot i_\infty=\overline i_\infty\cdot h_\infty$. This gives us the derived factorization
$$(\overline e^\dag\cdot h)\cdot i_\infty=\overline e^\dag \cdot \overline i_\infty\cdot h_\infty=\bot\cdot !\cdot i_\infty=\bot\cdot !$$

\item Parameter identity. Given $e:X\to SX$, with $X$ finitely presentable, and a morphism $h:Y\to SZ$, then the corresponding morphism of free Eilenberg-Moore algebras
$$\hat h=\mu_Z\cdot Sh:SY\to SZ$$
makes the triangle
$$\xymatrix{
X\ar[r]^-{e_Z^\dag}\ar[d]_{e^\dag_Y}&SZ\\
SY\ar[ru]_-{\hat h}
}$$
commutative. Indeed, $\hat h\cdot e^\dag_Y$ is strict because $e^\dag_Y$ is strict: from the strictness of $S$ we get $\hat h\cdot \bot=\bot$, thus
$$\hat h\cdot e^\dag_Y\cdot i_\infty=\hat h\cdot \bot\cdot !=\bot\cdot !$$
And $\hat h\cdot e^\dag_Y$ is a solution of $e$ in $SZ$:
$$\xymatrix{
X\ar[r]^{e_Y^\dag}\ar[dd]_-e&SY\ar[r]^-{\hat h}&SZ\\
&SSY\ar[u]_-{\mu_Y}\ar[rd]^-{S\hat h}\\
SX\ar[ru]^-{Se_Y^\dag}\ar[rr]_-{S(\hat h\cdot e^\dag_Y)}&&SSZ\ar[uu]_-{\mu_Z}
}$$

\item Double iteration identity. For every $e:X\to SX$, form $\hat e: SX \to SX$, then we have $$e^\dag=(\hat e\cdot e)^\dag$$
Indeed, $e^\dag$ is a solution of $\hat ee=\mu_X\cdot Se\cdot e$:
$$\xymatrix{
X\ar[rr]^-{e^\dag}\ar[d]_-e&&SY\\
SX\ar[r]^{Se^\dag}\ar[d]_-{Se}&SSY\ar[ru]^-{\mu_Y}\\
SSX\ar[r]^{SSe^\dag}\ar[d]_-{\mu_X}&SSY\ar[u]_-{S\mu_Y}\ar[rd]^-{\mu_{SY}}\\
SX\ar[rr]_-{Se^\dag}&&SSY\ar[uuu]_-{\mu_Y}
}$$
The upper part states $e^\dag$ is a solution, and the middle part
follows. The right-hand part is a monad axiom of $\mathbb S$ and the
lower part is the naturality of $\mu$. The strictness of $e^\dag$
w.r.t. $\hat e\cdot e$ is trivial: since $\mathbb S$ is an ideal monad on a hyper-extensive category,
we have the following diagram of pullback squares:
$$\xymatrix{
X_2\ar[r]^-{e_2}\ar[d]_-{i_2}&X_1\ar[r]^-{e_1}\ar[d]_-{i_1}&X\ar@{=}[r]\ar[d]_-{\eta_X}&X=X_0\ar[dd]^-{\eta_X}\\
X_1\ar[r]^-{e_1}\ar[d]_-{i_1}&X\ar[d]_-{\eta_X}\ar[r]^-e&SX\ar[d]_{\eta_{SX}}\\
X\ar[r]_-e&SX\ar[r]_-{Se}&SSX\ar[r]_-{\mu_X}&SX
}$$
So the first derived subobject of $\hat e\cdot e$ is $X_{2}$, and
similarly the $n$-th one is $X_{2n}$, where the $X_i$
are the derived subobjects of $e$. If follows that $X_\infty$ is the same for $e$ and $\hat e\cdot e$, thus, $e^\dag$ is a strict solution of $\hat e\cdot e$.

Observe that the above ``double iteration" extends to ``$n$ times iteration", e.g. for $n=3$ we get $e^\dag=(\hat e\hat e e)^\dag$.

(e) Dinaturality.
Given morphisms $f:X\to SZ$ and $g:Z\to SX$ with $X$ and $Z$ finitely presentable, 
form equation morphisms
$$\hat g\cdot f:X\to SX\qquad\text{and}\qquad\hat f\cdot g:Z\to SZ.$$
Their solutions are related by the dinaturality equation
$$(\hat g\cdot f)^\dag=\widehat{(\hat f\cdot g)^\dag}\cdot f:X\to SY$$
for every object $Y$.

Indeed, the right-hand side morphism is a solution of $\hat g\cdot f$ since it is the composite $\mu_Y\cdot S(\hat f\cdot g )^\dag\cdot f$, and the following diagram
$$\xymatrix{
X\ar[r]^-f\ar[d]_-f&SZ\ar[r]^{S(\hat f\cdot g )^\dag}\ar@{=}[dl]\ar[d]^-{Sg}&SSY\ar[r]^\mu&SY\\
SZ\ar[d]_-{Sg}\ar@{=}[ur]&SSX\ar[d]^-{SSf}\\
SSX\ar[d]_-{\mu}\ar[r]_-{SSf}\ar@{=}[ur]&SSSZ\ar[d]^-{S\mu}\\
SX\ar[r]_-{Sf}&SSZ\ar[r]_-{SS(\hat f\cdot g )^\dag}&SSSY\ar[r]_-{S\mu}\ar[uuu]_-{S\mu}&SSY\ar[uuu]_-\mu
}$$
commutes: the middle part follows from $(\hat f\cdot g )^\dag$ solving the equation morphism $\hat f\cdot g=\mu_Z\cdot Sf\cdot g$.

To see that $\widehat{(\hat f\cdot g)^\dag}\cdot f$ is a strict solution of $\hat g\cdot f$, we first need to relate the derived subobjects of $\hat g\cdot f$ and $f\cdot \hat g$. For this consider the two following chains of pullbacks:
\[\xymatrix{
\ar@{}[d]|{\objectstyle\cdots}&X_3\ar@{>->}[r]^-{i_3}\ar[d]_-{f_3}&X_2\ar@{>->}[r]^-{i_2}\ar[d]_-{f_2}&X_1\ar@{>->}[r]^-{i_1}\ar[d]_-{f_1}&X=X_0\ar[d]_-{f}\\
Z_3\ar@{>->}[r]_-{j_3}&Z_2\ar@{>->}[r]_-{j_2}&Z_1\ar@{>->}[r]_-{j_1}&Z\ar[r]_-{\eta_Z}&SZ
}
\]
\[
\xymatrix{
\ar@{}[d]|{\objectstyle\cdots}&Z_3\ar@{>->}[r]^-{j_3}\ar[d]_-{g_3}&Z_2\ar@{>->}[r]^-{j_2}\ar[d]_-{g_2}&Z_1\ar@{>->}[r]^-{j_1}\ar[d]_-{g_1}&Z=Z_0\ar[d]_-{g}\\
X_3\ar@{>->}[r]_-{i_3}&X_2\ar@{>->}[r]_-{i_2}&X_1\ar@{>->}[r]_-{i_1}&X\ar[r]_-{\eta_X}&SX
}
\]
Using them and the fact that $\mathbb{S}$ is an ideal monad in
connection with (hyper-)extensivity we can compute the derived
subobjects of $\hat g\cdot f$ as follows (all squares in following the diagram are pullbacks):
\[
\xymatrix{
&&X_4\ar[d]_-{f_4}\ar@{>->}[r]^-{i_4}&X_3\ar@{>->}[r]^-{i_3}\ar[d]_-{f_3}&X_2\ar@{>->}[r]^-{i_2}\ar[d]_-{f_2}&
X_1\ar@{>->}[r]^-{i_1}\ar[d]_-{f_1}&X\ar[d]^-f\\
&\cdots&Z_3\ar@{>->}[r]_-{j_3}\ar[d]_-{g_3}&Z_2\ar@{>->}[r]_-{j_2}\ar[d]_-{g_2}&Z_1\ar@{>->}[r]_-{j_1}\ar[d]_-{g_1}&Z\ar[d]_-{g}
\ar[r]_-{\eta_Z}&SZ\ar[d]^-{Sg}\\
X_4\ar@{>->}[r]_-{i_4}\ar@{=}[d]&X_3\ar@{=}[d]\ar@{>->}[r]_-{i_3}&X_2\ar@{=}[d]\ar@{>->}[r]_-{i_2}&X_1\ar@{=}[d]\ar@{>->}[r]_-{i_1}&X\ar@{=}[d]
\ar[r]_-{\eta_X}&SX\ar[r]_-{S\eta_X}&SSX\ar[d]^-{\mu_X}
\\
X_4\ar@{>->}[r]_-{i_4}&X_3\ar@{>->}[r]_-{i_3}&X_2\ar@{>->}[r]_-{i_2}&X_1\ar@{>->}[r]_-{i_1}&X\ar[rr]_-{\eta_X}&&SX\ar@{<-}`r[u]`[uuu]_-{\hat g\cdot f}[uuu]
}
\]
So we see that the first derived subobject of $\hat g\cdot f$ is
$i_1\cdot i_2:X_2\rightarrowtail X$, and the second one is $i_1\cdot
i_2\cdot i_3\cdot i_4:X_4\rightarrowtail X$ etc. similarly the $n$-th
derived subobject of $\hat f\cdot g$ is  
$$j_1\cdot j_2\cdot \cdots\cdot j_{2n}:Z_{2n}\rightarrowtail Z.$$
Now recall from \cite[Lemma~6.4]{abmv_how} that the intersections $X_\infty$ and
$Z_\infty$ are obtained after finitely many steps in the computation
of the derived subobjects, i.\,e.\ there exists an $n$ such that
$X_\infty=X_k$ and $Z_\infty =Z_k$ for all $k\geq n$. Thus, we
establish that $\widehat{(\hat f\cdot g)^\dag}\cdot f$ is a strict solution of $\hat g\cdot f$ by the commutative diagram below:
$$\xymatrix{
X_\infty=X_{n+1}\ar[r]^-{f_n}\ar[dd]_-{i_\infty}&Z_\infty=Z_{n}\ar[r]^-{!}\ar[d]_-{j_\infty}&1\ar[d]_-\bot\ar[rdd]^-\bot
\ar@{<-}`u[l]`[ll]_-{!}[ll]
\\
&Z\ar[r]_-{(\hat f \cdot g)^\dag}\ar[d]_-{\eta_Z}&SZ\ar[d]_-{\eta_{SZ}}\ar@{=}[rd]\\
X\ar[r]_-f&SZ\ar[r]_-{S(\hat f \cdot g)^\dag}&SSZ\ar[r]_-{\mu_Z}&SZ
}$$
This completes the proof.
\end{enumerate}
\end{exa}

\section{Conclusions and Open problems}

For coalgebras, recursivity caN be defined by the existence of unique algebra-to-coalgebra homomorphisms (no parameters are used). Or, equivalently, assuming the given endofunctor preserves weak pullbacks, by the unique solutions of all recursive systems with parameters. In contrast, in the dual situation we need to study non-equivalent variations. The present paper is dedicated to corecursive algebras $A$, where corecursivity means that  every recursive system of equations represented by a coalgebra has a unique solution in $A$. The formulation above is strictly weaker than the concept of a completely iterative algebra, where every parametrized recursive system of equations has a unique solution. For example, if we consider the endofunctor $X\mapsto X\times X$ of one binary operation in {\bf Set}, the algebra of all binary trees with finitely many leaves is corecursive, but not completely iterative.

The main result of our paper is the description of the free
corecursive algebra on $Y$ as the coproduct $MY=T\oplus FY$ of the
terminal coalgebra $T$ and the free algebra $FY$ in the category of all
algebras. The above example of binary trees is the free corecursive
algebra $M1$ on one generator. Our description is true for all
accessible ($=$ bounded)  endofunctors on {\bf Set} and, more generally, for
all endofunctors on {\bf Set} having free corecursive algebras. For
accessible monos-preserving endofunctors on more general base categories (posets, groups,
monoids etc.) the above description of the free corecursive algebras
also holds.

We introduce the concept of a corecursive monad, a weakening of
completely iterative monad. We prove that the assignment $Y\mapsto
MY=T\oplus FY$ is the free corecursive monad on the given accessible
endofunctor. And we characterize the Eilenberg-Moore algebras for this
monad. We call them Bloom algebras in honor of Stephen Bloom. They
play the analogous role that Elgot algebras, studied in \cite{amv3},
play for iterative monads: solutions of recursive equations are not
required to be unique, but have to satisfy some ``basic''
properties. In the case of Bloom algebras, the only property needed is
functoriality.

We further treat finitary equations: If we consider systems of
recursive equations as coalgebras $e:X\rightarrow HX$, then finite
systems of recursive equation are represented by coalgebras in which
$X$ is a finite set (or more generally, a finitely presentable
object). We can speak about finitary corecursive algebras as those in
which these finite systems have unique solutions. We prove that if $R$
is the initial iterative algebra, then $R \oplus FY$ is a free
finitary corecursive algebra.

Another question is: what is the analogy of the notion of an iteration
monad of S.~Bloom and Z.~\'{E}sik \cite{be} in the realm of
corecursive algebras? We do not know the answer. But at least we can
formulate the question precisely. The idea of iteration monads is to
collect all ``equational" properties that the operation $e\mapsto
e^\dagger $ of solving recursive systems $e$ has in trees for a
signature. This can be understood as forming the monad of free
iterative theories (or monads) on the category
$\mathbf{Set}^\mathbb{F}$ of sets in context, and characterizing
monadic algebras: these are, as proved in \cite{amv_em2},
precisely the iteration theories of S. Bloom and Z. \'{E}sik that are
functorial. So the open problem we state is this: form the monad of
free finitary corecursive theories on $\mathbf{Set}^\mathbb{F}$,
what are its monadic algebras? We called them Bloom monads, and listed
some of their properties.

\bibliography{ourpapers}
\bibliographystyle{plain}
\vspace{-40 pt}
\end{document}